\documentclass[10pt]{article}

\usepackage[english]{babel}
\usepackage[letterpaper]{geometry}
\usepackage[latin1]{inputenc}
\usepackage{fullpage}

\usepackage{tikz}
\usetikzlibrary{shapes,shapes.multipart, calc,matrix,arrows,arrows,positioning,automata}
\tikzset{
    >=stealth',
    punkt/.style={
           circle,
           rounded corners,
           draw=black, thick, 
           text width=2.5em,
           minimum height=3em,
           text centered},
    invisible/.style={
           draw=none,
           text width=1em,
           minimum height=3em,
           text centered},
    pil/.style={
           ->,
           thick,
           shorten <=2pt,
           shorten >=2pt,}
}

\usepackage{amsmath}
\usepackage{mathtools,amssymb,latexsym} 
\usepackage{amsxtra} 
\usepackage[mathscr]{eucal}
\usepackage{tikz-qtree}

\usepackage{color}

\begin{document}
\newtheorem{fact}{Fact}
\newtheorem{defin}{Definition}   
\newtheorem{examp}{Example} 
\newenvironment{nota}{\textbf{\flushleft Notation.}}{\bigskip}
\newtheorem{remark}{Remark}

\newtheorem{theorem}{Theorem}
\newtheorem{lemma}{Lemma}
\newtheorem{corollary}{Corollary}

\newenvironment{proof}[1][Proof]{\begin{trivlist}
\item[\hskip \labelsep {\bfseries #1}]}{\end{trivlist}}

\newcommand{\yst}{\,|\,} 
\newcommand{\yfun}[3]{#1:#2\rightarrow #3} 
\newcommand{\ywb}[1]{\rule{0pt}{2.35ex}\overline{#1}} 

\newcommand{\yia}{\mathcal{I}} 
\newcommand{\yias}{\yia^{\star}} 
\newcommand{\yoa}{\mathcal{O}} 
\newcommand{\yfsms}{\mathcal{M}} 
\newcommand{\yoas}{\yoa^{\star}} 
\newcommand{\yoasn}{\yoa'^{\star}} 
\newcommand{\ydel}{\delta} 
\newcommand{\ylam}{\lambda} 
\newcommand{\ydom}{D} 
\newcommand{\ymu}{\mu} 
\newcommand{\ytau}{\tau} 
\newcommand{\ymfm}{M=(S,s_0,\yia,\yoa,\ydom,\ydel,\ylam)} 
\newcommand{\ymfn}{N=(Q,q_0,\yia,\yoa',\ydom ',\ymu,\ytau)}   
\newcommand{\ydio}[2]{\mathrel{\overset{#1/#2}{\rightarrow}}} 
\newcommand{\ydi}[1]{\mathrel{\overset{#1/}{\rightarrow}}} 
\newcommand{\ydo}[1]{\mathrel{\overset{/#1}{\rightarrow}}} 
\newcommand{\yd}{\mathrel{\rightarrow}} 
\newcommand{\yddio}[3]{\mathrel{\substack{#1/#2\\\rightarrow\\#3}}}
\newcommand{\yddi}[2]{\mathrel{\substack{#1/\\\rightarrow\\#2}}}
\newcommand{\yddo}[2]{\mathrel{\substack{/#1\\\rightarrow\\#2}}}
\newcommand{\ydd}[1]{\mathrel{\underset{#1}{\rightarrow}}}
\newcommand{\yssim}{\mathrel{\prec}} 
\newcommand{\ysim}{\mathrel{\preccurlyeq}} 
\newcommand{\yequ}{\approx} 
\newcommand{\yalike}{\sim} 
\newcommand{\ynequ}{\not\yequ} 
\newcommand{\ynalike}{\not\sim} 

\newcommand{\yeps}{\varepsilon} 
\newcommand{\yfim}{\hfill$\Box$} 
 
\newcommand{\ysse}{\subseteq}
\newcommand{\ysdif}{\setminus}
\newcommand{\yemp}{\emptyset}
\newcommand{\ypow}[1]{\mathcal{P}(#1)}
\newcommand{\ysyd}{\ominus} 
\newcommand{\ymset}[1]{\mathcal{T}(#1)} 
\newcommand{\ymsett}[2]{\mathcal{T}_{#2}(#1)} 
\newcommand{\ymsetl}{\mathcal{L}} 

\title{Intrinsic Properties of Complete Test Suites}

\author{Adilson Luiz Bonifacio\thanks{Computing Department, University of Londrina, Londrina, Brazil. In collaboration with Computing Institute - UNICAMP.} \and
Arnaldo Vieira Moura\thanks{Computing Institute, University of Campinas, Campinas, Brazil.}}

\date{} 

\maketitle


\begin{abstract}
Completeness is a desirable property of test suites.
Roughly, completeness guarantees that a non-equivalent implementation under test will always be identified.
Several approaches proposed sufficient, and sometimes also necessary, conditions on the
specification model and on the test suite in order to guarantee completeness.
Usually, these approaches impose several restrictions on the specification and on 
the implementations, such as requiring them to be reduced or
complete.
Further, test cases are required to be non-blocking~---~that is, they must run to completion~---~ on both the specification and the
implementation models.
In this work we deal  test cases that can be blocking,
we define a  new notion that captures completeness, and we characterize test suite completeness in this new scenario. 
We establish an upper bound on the number of states of implementations beyond which no test suite can be complete, both in the classical sense and in the new scenario with 
blocking test cases.
\end{abstract}

\section{Introduction}

Completeness of test suites has been largely studied for models based on Finite State Machines (FSMs)~\cite{BoniMS-Model-2012,HierU-Reduced-2002,DoroEY-Improved-2005,simap-checking-2010,BonifacioSAC2014,UralWZ-Minimizing-1997,SimaPY-Reducing-2012}.
A test suite  is called complete for a FSM specification when it provides  complete fault coverage~\cite{BoniMS-Model-2012,HierU-Reduced-2002}. 
Several works have proposed strategies for generating complete test suites~\cite{SilvPY-Generating-2009}, or for checking if a given test suite is complete for a given specification~\cite{BonifacioSAC2014}. 
Some of them presented necessary conditions~\cite{PetrB-Fault-1996,YaoPB-Fault-1994} for test suite completeness, whereas other approaches gave sufficient, but not necessary, conditions for test suite completeness~\cite{DoroEY-Improved-2005,PetrY-test-2000,simap-checking-2010,UralWZ-Minimizing-1997}.
Some more recent works have described necessary \emph{and} sufficient conditions for test suite completeness~\cite{BonifacioSAC2014,SilvPY-Generating-2009}.
All these works imposed  restrictions on the specification and implementations, or over the fault domains~\cite{DoroEY-Improved-2005,PetrY-test-2000,simap-checking-2010,UralWZ-Minimizing-1997,BonifacioSAC2014}. 
Some of them considered specifications with $n$ states and restricted  the implementations under test to have at most $n$ states. 
Further, in some approaches specification and implementations are required to be reduced  or completely specified machines. 
Always, test cases have been required to be non-blocking on both the specifications and the implementations models. This meaning that all test cases are assumed to run
to completion in these models. In particular, even if implementations are treated as black boxes, all test cases are assumed to run to completion on implementations. 

In this work we deal with the more general scenario where test cases can be blocking.
In particular, we do not require that all test cases run to completion when implementations can be \emph{partial FSMs}, and are treated as \emph{true black boxes}. 
We propose a new notion of equivalence, called ``alikeness'', and  
we extend the classical notion of equivalence when blocking test cases can be present, thus giving rise to the notion of ``perfectness'', \emph{in lieu} of the classical 
notion of completeness. 
We then use bi-simulation relations and reducibility over machines to characterize test suite perfectness in this new more general scenario.
  
A related issue that concerns test suite completeness is the maximum size of implementations that can be put under test.
Usually, earlier works constrained implementations to have at most the same number of states as the  given specification.
We are not aware for any work that gives a precise relationship between the maximum number of states in implementations and the size of test suites in order to get positive verdicts when such
implementations are put under  test. 
Here, we establish a precise upper bound on the number of states of implementations under test, beyond which no test suite can be complete, both in the classical sense and in the more general scenario when blocking test cases can be present.
The bound is based on test suite size and the number of states in the given specification.

We organize the paper as follows.
Basic definitions and notations appear in Section~\ref{basicconcepts}. 
Section~\ref{isomorphism}  gives the perfectness of test suites in terms of the property of isomorphism between machines. 
We relate the well-known notion of completeness to the notion of perfectness in Section~\ref{completenessxperfectness}. 
In Section~\ref{bounding-impls} we establish an upper bound on the number of states in  candidate implementations beyond which no test suite is complete. 
Section~\ref{m-perfectness} defines the notion of $m$-perfectness, where $m$ is the number of candidate implementations. 
Section~\ref{conclusion} states some conclusions. 


\section{Definitions and notation}\label{basicconcepts}

Let $\yia$ be an alphabet.
The length of any finite sequence $\alpha$ of symbols  over $\yia$ is indicated by $|\alpha|$.
The empty sequence will be indicated by $\yeps$, with $|\yeps|=0$.
The set of all sequences of length $k$ over $\yia$ is denoted by $\yia^k$, while
$\yias$ names the set of all finite sequences over $\yia$.
When we write $\sigma=x_1x_2\cdots x_n\in\yias$ ($n\geq 0$) we mean $x_i\in\yia$ ($1\leq i\leq n$),
unless noted otherwise, and similarly for other alphabets. 
Given any two sets of sequences $A, B\ysse \yias$, their symmetric difference will be indicated by
$A\ysyd B$, that is  $A\ysyd B=(\overline{A}\cap B)\cup (A\cap\overline{B})$,
where $\overline{A}$ indicates the complement of $A$ with respect to $\yias$.
The usual set difference is indicated by $A\ysdif B$. 

\begin{remark}\label{rem:dif-sim}
$A\ysyd B=\yemp$ iff\,%
\footnote{ Here, `iff' is short for `if and only if'.}
$A=B$.
\end{remark}

\subsection{Finite state machines and test suites}
Next, we write the definition of a Finite State Machine~\cite{BonifacioSAC2014,Gill-Introduction-1962}. 
\begin{defin}\label{fsm}
A FSM is a system $\ymfm$ where
\begin{itemize}
\item $S$ is a finite set of \emph{states}
\item $s_0\in S$ is the \emph{initial state}
\item $\yia$ is a finite set of \emph{input actions} or \emph{input events}
\item $\yoa$ is a finite set of \emph{output actions} or \emph{output events}
\item $D\subseteq S\times\yia$ is a \emph{specification domain}
\item $\delta: D\rightarrow S$ is the \emph{transition function}
\item $\lambda: D\rightarrow \yoa$ is the \emph{output function}. 
\end{itemize}
\end{defin}

In what follows $M$ and $N$ will always denote the FSMs 
$(S,s_0,\yia,\yoa,\ydom,\ydel,\ylam)$ 
and $(Q,q_0,\yia,\yoa',\ydom ',\ymu,\ytau)$, respectively.
Let $\sigma=x_1x_2\cdots x_n\in\yias$,
$\omega=a_1a_2\cdots a_n\in\yoas$ ($n\geq 0$). 
If there are states $r_i\in S$ ($0\leq i\leq n)$ such that
$\delta(r_{i-1},x_i)=r_i$ and  $\lambda(r_{i-1},x_i)=a_i$ ($1\leq i\leq n$), then
we may write  $r_0\ydio{\sigma}{\omega} r_n$.
When the input sequence $\sigma$, or the output sequence $\omega$, is not 
important, then we may write $r_0\ydi{\sigma} r_n$, or $r_0\ydo{\omega} r_n$, respectively, and 
when both sequences are not important we may write $r_0\yd r_n$.
We can also drop the target state, and write
$r_0\ydio{\sigma}{\omega}{}$ or $r_0\yd {}$.
It will be useful to extend the functions $\delta$ and $\lambda$ to pairs 
$(s,\sigma)\in  S\times \yias$.
Let $\widehat{D}=\Big\{(s,\sigma)\,\Big|\, s\ydi{\sigma}, \sigma\in \yias, s\in Q \Big\}$.
Define the extensions $\yfun{\widehat{\delta}}{\widehat{D}}{S}$ and 
$\yfun{\widehat{\lambda}}{\widehat{D}}{\yoas}$ by letting
$\widehat{\delta}(s,\sigma)=r$ and $\widehat{\lambda}(s,\sigma)=\omega$
whenever $s\ydio{\sigma}{\omega} r$.
When there is no reason for confusion, we may write $D$,
$\delta$ and $\lambda$ instead of $\widehat{D}$, $\widehat{\delta}$ and $\widehat{\lambda}$,
respectively. 
Also, the function $\yfun{U}{S}{\yias}$ will be useful, where $U(s)=\{\sigma\yst (s,\sigma)\in\widehat{D}\}$.
Informally, $U(s)$ denotes all input action sequences that can be run from the state $s$. 

Now we are in a position to define test cases and test suites. 
\begin{defin}\label{defin:teste-suite}
Let $M$ be a FSM.  
A \emph{test suite for $M$} is any finite nonempty subset of $\yias$.
Any element of a test suite is a \emph{test case}. 
\end{defin}

Before we can define test completeness, we need the classical notions of distinguishability and equivalence. 

\begin{defin}\label{disting-equiv-partial}
Let $M$ and $N$  be FSMs and let $s\in S$, $q\in Q$.
Let $C\ysse \yias$.
We say that  $s$ and $q$ are
\emph{$C$-distinguishable} iff $\lambda(s,\sigma)\neq \tau(q,\sigma)$ for some $\sigma\in U(s)\cap U(q)\cap C$, denoted $s\ynequ_C q$.
Otherwise, $s$ and $q$ are \emph{$C$-equivalent}, denoted $s\yequ_C q$.
We say that $M$ and $N$ are \emph{$C$-distinguishable} 
iff  $s_0\ynequ_C q_0$, and they are \emph{$C$-equivalent} iff $s_0\yequ_C q_0$.
\end{defin}
When $C$ is not important, or when it is clear from the context, we might drop the index.
When there is no mention to $C$, we understand that we are taking $C=\yias$.
In this case, the condition $U(s_0)\cap U(q_0)\cap C$ reduces to $U(s_0)\cap U(q_0)$.
For the ease of notation, we also write $M\yequ_C N$ when $M$ and $N$ are $C$-equivalent, and  
$M\ynequ_C N$ when they are $C$-distinguishable.

Now we can state the conventional notion  of a $m$-complete test suite.
\begin{defin}\label{def:complete-partial}
 Let $M$ be a FSM and $T$ a test suite for $M$. Let $m\geq 1$.
Then $T$ is \emph{$m$-complete for $M$} iff for any FSM $N$, with $U(s_0)\subseteq U(q_0)$ and with at most $m$ states, if $M\ynequ N$ then $M\ynequ_T N$. 
 \end{defin}
Note that if $\sigma$ runs to completion from $s_0$, that is, 
$s_0\ydi{\sigma}$, then $\sigma$ must also run to completion from $q_0$, that is we must have $q_0\ydi{\sigma}$.
The definition says that any discrepancy between the behaviors of the specification $M$ and any implementation $N$ will be detected if we run the tests in $T$ through $M$ and $N$, 
provided that we consider implementations with at most $m$ states.
Note that the technical condition  $U(s_0)\subseteq U(q_0)$ will always be satisfied if
we were to test implementations that were complete FSM models.
A FSM $M$ is said to be complete when $D=S\times \yia$, that is, for any state $s$ and any 
input symbol $x$, we always have $s\ydi{x} {}$.

\subsection{The notion of `alikeness'}

A \emph{blocking} test case for $M$ is a sequence $\sigma\not\in U(s_0)$, otherwise we say that $\sigma$ \emph{runs to completion} in $M$.
Then, given two FSM models $M$ and $N$, if $\sigma\in U(s_0)\ysyd U(q_0)$,
either $\sigma$ blocks in $M$ and runs to completion in $N$, or vice-versa. 
Given a test suite $T$ and two FSM models $M$ and $N$ ,we want to say when $M$ and $N$ are equivalent in some more general sense, that is, even considering that we may have blocking test cases, for $M$ or $N$, in $T$.
Intuitively, all $\sigma\in T$ that is a blocking test case for $M$ must also be
 a  blocking test case for $N$, and vice-versa.
Furthermore, any test case that is non-blocking for both $M$ and $N$ must  output identical behaviors when run through both models.
In this case $M$ and $N$ will be said to be \emph{$T$-alike}.

\begin{defin}
\label{disting-equiv-block}
Let $M$ and $N$  be  FSMs and let $s\in S$, $q\in Q$.
Let $C\ysse \yias$.
We say that  $s$ and $q$ are \emph{$C$-alike},  denoted $s\yalike_C q$, iff $\big(U(s)\ysyd U(q)\big)\cap C = \yemp$ and $\lambda(s,\sigma)= \tau(q,\sigma)$ for all $\sigma\in U(s)\cap U(q)\cap C$.
Otherwise, $s$ and $q$ are \emph{$C$-unlike}, denoted $s\ynalike_C q$.
We say that $M$ and $N$ are \emph{$C$-alike} iff  $s_0\yalike_C q_0$, otherwise they are \emph{$C$-unlike}. 
\end{defin}

We may also write $M\yalike_C N$ when $M$ and $N$ are $C$-alike, or $M\ynalike_C N$ when they are
$C$-unlike.
Again, when $C$ is not important, or when it is clear from the context, we might drop the index, and when there is no mention to $C$, we understand that we are taking $C=\yias$.

\begin{remark}\label{rema:equal-Us}
We note of the following simple observations. 
\begin{enumerate}
\item Using Remark~\ref{rem:dif-sim}, we note that $s\yalike q$ is equivalent to $U(s)=U(q)$ and $\lambda(s,\sigma)= \tau(q,\sigma)$ for all $\sigma\in U(s)$. 
\item If $C_1\ysse C_2$, then $s\yalike_{C_2} q$ implies $s\yalike_{C_1} q$.
\item If $s\yalike q$, then $s\yalike_C q$, for all $C\ysse\yias$. \yfim
\end{enumerate}
\end{remark}

An important aspect of the alikeness relation, $\yalike_C $, is that it is an equivalence relation when $M$ and $N$ are the same machine, that is, when $\yalike_C$ is defined over 
a single set.
We note that this is not the case, in general, with the distinguishability relation $\yequ_C$.

\begin{lemma}\label{lem:alike-trans}
Let $M$ be an FSM and let $C\ysse \yias$.
Then $\yalike_C$ is an equivalence relation on $S$.
\end{lemma}
\begin{proof}
Let $s, r, p\in S$ be states of $M$.
We clearly have $U(s)\ysyd U(s)=\yemp$ and $\lambda(s,\alpha)=\lambda(s,\alpha)$ for all $\alpha\in U(s)\cap C$.
So, $\yalike_C $ is reflexive.
Also, set intersection, the symmetric set difference $\ysyd$ and, of course, equality are
commutative. Hence, $\yalike_C $ is symmetric.

For transitivity, assume $s\yalike_C r$ and $r \yalike_C p$. Let $\alpha\in U(s)\cap C$.
Thus $\alpha\in U(r)$ because  $s\yalike_C r$, and then $\alpha\in U(p)$ because  $r\yalike_C p$. So, $U(s)\ysse U(p)$. 
Since we already have symmetry, we get $p \yalike_C r$ and $r \yalike_C s$, and a similar argument gives $U(p)\ysse U(s)$, showing that $(U(s)\ysyd U(p))\cap C=\yemp$.
Now, let $\alpha\in U(s)\cap U(p)\cap C$. 
Since $s\yalike_C r$, we get $\alpha\in U(r)$ and so $\lambda(s,\alpha)=\lambda(r,\alpha)$.
But also $r\yalike_C p$, and so $\lambda(r,\alpha)=\lambda(p,\alpha)$, thus establishing
$\lambda(s,\alpha)=\lambda(p,\alpha)$.
We may  then conclude that $s\yalike_C p$, and $\yalike_C$ is transitive. \yfim
\end{proof}

\begin{remark}\label{rema:alike-trans}
We note that, in Lemma~\ref{lem:alike-trans}, the argument establishing the transitivity
of the alikness relation $\yalike_C$ is still valid when it is defined as a relation between
the states of two distinct machines. 
\end{remark}

When reducing FSMs in the presence of blocking test cases, we will need the following
technical result.

\begin{lemma}\label{lem:alike}
Let $M$ be a FSM and let $s,r\in S$ be states of $S$, with $s\yalike r$.
\begin{list}{}{\setlength{\itemsep}{0pt}\setlength{\parsep}{0pt}\setlength{\topsep}{0pt}\setlength{\parskip}{0pt}}
\item[\rm(1)] If $s\ydio{x}{a} p$ with $x\in\yia$ and $a\in\yoa$, then $r\ydio{x}{a} q$ with $p\yalike q$, for some $q\in S$.
\item[\rm(2)] If $s\ydio{\alpha}{\omega} p$ with $\alpha\in\yias$ and $\omega\in\yoas$, then $r\ydio{\alpha}{\omega} q$, with $p\yalike q$
for some $q\in S$.
\end{list}
\end{lemma}
\begin{proof}
We first treat item 1.
We have $x\in U(s)$, and so $x\in U(r)$ because $s\yalike r$, which leads to $r\ydio{x}{b} q$ for some $q\in S$, $b\in\yoa$. 
Now, $x\in U(s)\cap U(r)$ and, since $s\yalike r$, we get $a=\lambda(s,x)=\lambda(r,x)=b$.
It remains to show that $p\yalike q$. 
Let $\alpha\in U(p)$. Then $x\alpha\in U(s)$, and again $x\alpha\in U(r)$. 
Since $M$ is deterministic, this gives $\alpha\in U(q)$, and so $U(p)\ysse U(q)$. 
Using Remark~\ref{rema:equal-Us}(1) we have $r\yalike s$, and a similar argument gives $U(q)\ysse U(p)$.
We conclude that $U(p)=U(q)$, and so $U(p)\ysyd U(q)=\yemp$.
Now, let $\beta\in U(p)\cap U(q)$. Then, $x\beta\in U(s)\cap U(r)$, and since $s\yalike r$
this gives 
$a\lambda(p,\alpha)=\lambda(s,x\beta)=\lambda(r,x\beta)=a\lambda(q,\alpha)$.
We conclude that $\lambda(p,\alpha)=\lambda(q,\alpha)$, as desired.

Now, item (2) follows by a simple indiction on $|\alpha|\geq 0$, and using the result of item 1. \yfim
\end{proof}

The notion of \emph{perfectness} has been introduced by Bonifacio and Moura~\cite{BonifacioSEFM2014,TR-IC-13-33}, in order to cope with test cases that may not run to completion either in the specification or in  the implementation models.  
It is based on the notion of alikness.
\begin{defin}[\cite{BonifacioSEFM2014}] 
\label{def:n-complete-relax}
Let $M$ be a FSM and $T$ be a test suite for $M$.
Then $T$ is \emph{perfect for $M$} iff for any FSM $N$, if $M\ynalike N$ then $M\ynalike_T N$.
\end{defin}
That is, when $T$ is a perfect test suite for a specification $M$, then for any implementation under test $N$, if $M$ and $N$ are unlike, then they are also $T$-unlike.

In Definition~\ref{def:n-complete-relax}, there is no limit in the size of the implementations. 
In the next definition, the key property of $M\ynalike N$ implying $M\ynalike_T N$ is
required to hold only for implementations with up to a predefined number of states.

\begin{defin} \label{def:n-perfecteness}
Let $M$ be a FSM, let $T$ be a test suite for $M$, and let $m\geq 1$.
Then $T$ is \emph{$m$-perfect for $M$} iff for any FSM $N$ with at most $m$ states, if $M\ynalike N$ then $M\ynalike_T N$.
\end{defin}

\subsection{Simulations and perfectness}
In~\cite{BonifacioSEFM2014,TR-IC-13-33} bi-simulation was used to 
\emph{characterize} test suite perfectness.

\begin{defin}\label{simulation}
Let $M$ and $N$  be FSMs.  
We say that a relation $R\ysse S\times Q$ is a \emph{simulation} (\emph{of $M$ by $N$}) iff  $(s_0,q_0)\in R$, and whenever we have $(s,q)\in R$ and $s\ydio{x}{a} r$ in $M$,
 then there is a state $p\in Q$ such that $q\ydio{x}{a} p$ in $N$ and with $(r,p)\in R$.
We say that $M$ and $N$ are \emph{bi-similar} iff there are simulation relations
$R_1\ysse S\times Q$ and $R_2\ysse Q\times S$. 
\end{defin}

The following simple facts will be used later.

\begin{fact}\label{fact:sim-trans}
The simulation relation is transitive, that is, let 
$M_i=(S_i,s_i,\yia,\yoa,\ydom_i,\ydel_i,\ylam_i)$ be FSMs, $i=1,2,3$, and where $M_2$ simulates $M_1$ and $M_3$ simulates $M_2$. 
Then, $M_3$ simulates $M_1$.
\end{fact}
\begin{proof}
Let $R_1\ysse S_1\times S_2$ and $R_2\ysse S_2\times S_3$ be simulation relations.
Define $R\ysse S_1\times S_3$ by $(s,p)\in R$ iff $(s,q)\in R_1$ and $(q,p)\in R_2$, for some 
$q\in S_2$.
Firstly, since $(s_1,s_2)\in R_1$ and $(s_2,s_3)\in R_2$ we get $(s_1,s_3)\in R$, as needed.
Moreover, let $(s,p)\in R$ and $s\ydio{x}{a} s_1$. 
We must have $(s,q)\in R_1$ and $(q,p)\in R_2$ for some $q\in S_2$.
Since $R_1$ is a simulation, we get $q\ydio{x}{a} q_1$, with $(s_1,q_1)\in R_1$.
Since $R_2$ is a simulation, we get $p\ydio{x}{a} p_1$ with $(q_1,p_1)\in R_2$.
Then, $(s_1,p_1)\in R$, as desired.
\end{proof}

\begin{fact}\label{fact:simul} 
Let $M$ and $N$ be FSMs, and let $R\ysse S\times Q$ be a simulation of $M$ by $N$.
If $(s,q)\in R$ and $\delta(s,\alpha)=r$ for some $\alpha\in\yias$, then 
$\mu(q,\alpha)=t$ with $(r,t)\in R$, for a unique $t\in Q$.
\end{fact}
\begin{proof}
An easy induction on $|\alpha|\geq 0$. 
Such a $t\in Q$ is unique, since $N$ is deterministic. \yfim
\end{proof}

\begin{fact}\label{fact:bisimul} 
Let $M$ and $N$ be FSMs,  let $R\ysse S\times Q$ be a simulation of $M$ by $N$,
and let $L\ysse Q\times S$ be a simulation of $N$ by $M$.
Let $(s,q)\in R$, $(q,s)\in L$, and $\alpha\in\yias$.
If $\delta(s,\alpha)=r$, then 
$\mu(q,\alpha)=t$ with $(r,t)\in R$ and $(t,r)\in L$, for a unique $t\in Q$.
\end{fact}
\begin{proof}
From $\delta(s,\alpha)=r$ and $(s,q)\in R$ Fact~\ref{fact:simul} gives a unique $t\in Q$ with $\mu(q,\alpha)=t$
and $(r,t)\in R$.
From $(q,s)\in L$ and $\mu(q,\alpha)=t$, Fact~\ref{fact:simul} again gives some $p\in S$
with $(t,p)\in L$ and $\delta(s,\alpha)=p$. 
Since $M$ is deterministic and we already have $\delta(s,\alpha)=r$ we conclude that
$p=r$.
Hence,  $(t,r)\in L$ as desired. \yfim
\end{proof}

The next lemma shows a useful relationship between bi-simulations and alikeness.
\begin{lemma}\label{lema:RL}
Let $M$ and $N$ be FSMs,  let $R\ysse S\times Q$ be a simulation of $M$ by $N$,
and let $L\ysse Q\times S$ be a simulation of $N$ by $M$.
Let $(s_i,q)\in R$ and $(q,s_i)\in L$, $i=1,2$.
Then, $s_1\yalike s_2$.
\end{lemma}
\begin{proof}
For the sake of contradiction, assume that $s_1\ynalike s_2$.
Definition~\ref{disting-equiv-block} gives some $r_i\in S$, $a_i\in\yoa$ ($i=1,2$),
$x\in \yia$, and 
some $\alpha\in\yias$ with $s_i\ydi{\alpha} r_i$ ($i=1,2$),  and such that 
for some $t_1,t_2\in  S$, either
\begin{list}{}{\setlength{\itemsep}{0pt}\setlength{\parsep}{0pt}\setlength{\topsep}{0pt}}
\item[(1)] $r_i\ydio{x}{a_i} t_i$, $i=1,2$, and $a_1\neq a_2$; or
\item[(2)] $r_1\ydio{x}{a_1} t_1$, and $x\not\in U(r_2)$; or 
\item[(3)] $r_2\ydio{x}{a_2} t_2$, and $x\not\in U(r_1)$. 
\end{list}

From $(s_i,q)\in R$ and $s_i\ydi{\alpha} r_i$, Fact~\ref{fact:simul} gives $u_i\in Q$ such that 
$q\ydi{\alpha} u_i$ and $(r_i,u_i)\in R$, for $i=1,2$.
Since $N$ is deterministic, we get $u_1=u_2=u$ and so $(r_i,u)\in R$ ($i=1,2$).  

Now, if case (1) holds, then from $(r_i,u)\in R$ and using  Definition~\ref{simulation} we get
$u\ydio{x}{a_i} v_i$ for some $v_i\in Q$ ($i=1,2$). Again, since $N$ is deterministic,
we obtain $a_1=a_2$, a contradiction.

Assume that case (2) holds. 
Since $(r_1,u)\in R$ and $r_1\ydio{x}{a_1} t_1$, Definition~\ref{simulation} gives 
$u\ydio{x}{a_1} v_1$, for some $v_1\in Q$.
From $q\ydi{\alpha} u_2$ and $(q,s_2)\in L$, Fact\ref{fact:simul} gives some $r'_2\in S$
with $s_2\ydi{\alpha} r'_2$ and $(u_2,r'_2)\in L$.
But we already have $s_2\ydi{\alpha} r_2$, and so the determinism of $M$ gives $r'_2=r_2$.
Hence, $(u_2,r_2)\in L$ and then $(u,r_2)\in L$ because $u_2=u$.
But we also have $u\ydio{x}{a_1} v_1$ and so, using Definition~\ref{simulation},
we get $x\in U(r_2)$, contradicting the hypothesis of case (2).

Case (3) also leads to a contradiction, by a reasoning entirely analogous as was done 
for case (2). 

We conclude that, in fact, $s_1\yalike s_2$, as desired. \yfim
\end{proof}

The following result establishes a necessary and sufficient condition for perfectness.
\begin{theorem}[\cite{BonifacioSEFM2014}]
\label{necessity-sufficiency-block}
Let $M$ be a FSM and $T$ be a test suite for $M$. 
Then $T$ is perfect for $M$ iff any $T$-alike FSM is bi-similar to $M$. 
\end{theorem}

In the next section we show that the bi-similarity test can be exchanged for an isomorphism test.

\section{Perfectness  and Isomorphism}\label{isomorphism}

In this section we characterize perfectness  in terms of isomorphisms between FSMs. 

\subsection{Bi-simulation and isomorphism}
Two FSMs are said to be isomorphic when they specify exactly the same model, except for a state relabeling.

\begin{defin}
Let $M$ and $N$ be FSMs with $\yoa=\yoa'$.
An \emph{isomorphism} (of $M$ into $N$) is a bijection $\yfun{f}{S}{Q}$ such that 
\begin{enumerate}
\item $f(s_0)=q_0$; and
\item $s\ydio{x}{a} r$ in $M$ if and only if $f(s)\ydio{x}{a} f(r)$ in $N$,
for all $x\in\yia$, $a\in\yoa$. 
\end{enumerate}
Machines $M$ and $N$ are \emph{isomorphic} iff there is an isomorphism of $M$ into $N$.\yfim
\end{defin}

\begin{remark}\label{remk:isomorf}
Let $M$ and $N$ be FSMs.
The following are immediate consequences:
\begin{enumerate}
\item $f$ is an isomorphism of $M$ into $N$ if and only if $f^{-1}$ is an isomorphism 
of $N$ into $M$.
\item Any isomorphism of $M$ into $N$ is also a simulation of $M$ by $N$.
\end{enumerate}
\end{remark}

The first half of the characterization is easily obtained.
\begin{lemma}\label{lema:iso-to-bisim}
Let $M$ and $N$ be isomorphic FSMs. Then, $M$ and $N$ are bi-similar.
\end{lemma}
\begin{proof}
Using Remark~\ref{remk:isomorf}, we have a simulation of $M$ by $N$, and vice-versa.\yfim
\end{proof}

Now let $M$ and $N$ be bi-similar. 
It is clear that if all states in $M$ are unlike, but $N$ has two distinct states that are alike, then it is possible for $M$ and $N$ not to be isomorphic, since these two distinct equivalent states in $N$ would have to correspond to a single state in $M$.
Machines illustrated in Figures~\ref{fsm-n1} and~\ref{fsm-spec} are a case in point. 
\begin{figure}[htb]
\center

\begin{tikzpicture}[node distance=1cm, auto,transform shape]
  \node[ inner sep=1pt, punkt] (q0) {$q_0$};
  \node[punkt, inner sep=1pt,right=2cm of q0] (q1) {$q_1$};
  \node[punkt, inner sep=1pt,right=2cm of q1] (q2) {$q_2$};

  \path (q0)   edge[pil,loop above] 
                	node[above]{$0/1$} (q0);
\path (q0)    edge [ pil, left=20]
                	node[anchor=north,above]{$1/1$} (q1);
\path (q1)    edge [ pil, bend left=20]
                	node[above]{$0/0$} (q2);
\path (q2)    edge [ pil, bend left=20]
                	node[below]{$0/0$} (q1);
\end{tikzpicture}
\caption{FSM $N_1$.}
\label{fsm-n1}
\end{figure} 
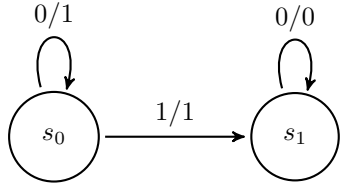
\begin{figure}[htb]
\center

\begin{tikzpicture}[node distance=1cm, auto,transform shape]
  \node[ punkt] (q0) {$s_0$};
  \node[punkt, inner sep=3pt,right=2cm of q0] (q1) {$s_1$};
  \path (q0)   edge[pil, loop above]
                	node[anchor=left,above]{$0/1$} (q0);
\path (q0)    edge [ pil]
                	node[anchor=north,above]{$1/1$} (q1);
\path (q1)    edge [pil,loop above] 
                	node[anchor=left,above]{$0/0$} (q1);
\end{tikzpicture}
\caption{Specification FSM $M$.}
\label{fsm-spec}
\end{figure}
The problem, of course, is that states $q_1$ and $q_2$ in $N_1$ have exactly the same 
blocking input sequences and, moreover, the 
behaviors of $q_1$ and $q_2$ in $N_1$ are exactly the same
under any input sequence $\sigma$ that is non-blocking for both of them.

In the classical sense, a FSM $M$ is reduced if every pair of distinct states in $S$ are distinguishable.
When treating partial FSM, however, we need also to take into consideration blocking input sequences. 
In order to differentiate from the classical notion of reduction in FSMs, we name reduction in the presence of blocking sequences as $p$-reduction.
Both definitions are very similar.

\begin{defin}
A FSM $M$ is \emph{reduced}
iff every pair of distinct states of $S$ are distinguishable, and
for all state $s\in S$ there is a $\sigma\in\yias$ with $\delta(s_0,\sigma)=s$. \yfim
\end{defin}

\begin{defin}\label{def:p-reduced}
A FSM $M$ is $p$-reduced iff any no two distinct states in  $M$ are alike and, moreover,
for all $s\in S$ there is $\alpha\in\yias$ with $\delta(s_0,\alpha)=s$. 
\yfim
\end{defin}
Hence, for any two distinct states $s$ and $r$ in $M$ there is an input sequence that is a 
blocking sequence for one of them and is not blocking for the other, or there is an 
input sequence that is non-blocking for both $s$ and $r$ but yields different behaviors
when starting at the two.  
Returning to Figures~\ref{fsm-n1} and~\ref{fsm-spec}, we see that the presence of $q_1$ and $q_2$ in $N_1$ shows that it is not a $p$-reduced FSM.
\begin{remark}\label{rem:reduced}
If $M$ is a reduced FSM with at least two reachable states, then
there always exists a transition out of any reachable state $s$, that is $(s,x)\in D$ for some $x\in\yia$.
Otherwise, $s$ could not be distinguished from any other reachable state in $M$.
\end{remark}

We proceed to show, by a series of simple facts, that if $M$ and $N$ are bi-similar and $p$-reduced, then they are isomorphic.
We start by noting that the bi-similarity condition gives two simulation relations $R\ysse S\times Q$ and $L\ysse Q\times S$.
Define a relation $f\ysse S\times Q$ as follows:
$$(s,q)\in f\quad\text{iff\quad $s_0\ydi{\alpha} s$ and $q_0\ydi{\alpha} q$, for some $\alpha\in\yias$}.$$

\begin{fact}\label{fact:RL}
If $(s,q)\in f$ then $(s,q)\in R$ and $(q,s)\in L$.
\end{fact}
\begin{proof}
Observe that $(s,q)\in f$ gives $s_0\ydi{\alpha} s$ and $q_0\ydi{\alpha} q$.
Since $(s_0,q_0)\in R$, Fact~\ref{fact:simul} gives $q_0\ydi{\alpha} p$ and $(s,p)\in R$, for some $p\in Q$. Since $N$ is deterministic, we get $p=q$, and so $(s,q)\in R$.
A symmetric argument gives $(q,s)\in L$. \yfim
\end{proof}

Now we show that $f$ is, in fact, a bijection. This will establish that $M$ and $N$ are isomorphic, when they are $p$-reduced.

\begin{description}
\item[\it $f$ is a function:] Let $(s,q_i)\in f$, $i=1,2$.
From Fact~\ref{fact:RL} we obtain $(s,q_i)\in R$ and $(q_i,s)\in L$, $i=1,2$. 
Using Lemma~\ref{lema:RL}, we conclude that $q_1\yalike q_2$.
Because $N$ is $p$-reduced, Definition~\ref{def:p-reduced} forces $q_1=q_2$. 

\item[\it $f$ is total:] Let $s\in S$. Since $M$ is $p$-reduced,
Definition~\ref{def:p-reduced} gives $\alpha\in\yias$ such that $s_o\ydi{\alpha} s$.
Since $(s_0,q_0)\in R$, Fact~\ref{fact:simul} gives $q_0\ydi{\alpha} q$ for some
$q\in Q$.
Thus, $(s,q)\in f$.

\item[\it $f$ is onto:] Let $q\in Q$. Since $N$ is $p$-reduced,
Definition~\ref{def:p-reduced} gives $\alpha\in\yias$ such that $q_o\ydi{\alpha} q$.
Since $(q_0,s_0)\in L$, Fact~\ref{fact:simul} gives $s_0\ydi{\alpha} s$ for some
$s\in S$.
Thus, $(s,q)\in f$.

\item[\it $f$ is one-to-one:] Let $(s_i,q)\in f$, $i=1,2$. 
Using Fact~\ref{fact:RL} we get $(s_i,q)\in R$ and $(q,s_i)\in L$, $i=1,2$. 
Then Lemma~\ref{lema:RL} gives  $s_1\yalike s_2$.
Thus $s_1=s_2$, since $M$ is $p$-reduced.

\item[\it $f$ is a bijection:] We have shown that $f$ is a total function, which is also
onto and injective.
\end{description}

We can now state the main result of this section.
\begin{theorem}\label{completeness-isomorphim}
Let $M$ and $N$ be $p$-reduced FSMs.
Then, $M$ and $N$ are bi-similar if and only if $M$ and $N$ are isomorphic.
\end{theorem}
\begin{proof}
If $M$ and $N$ are isomorphic then they are bi-similar by Lemma~\ref{lema:iso-to-bisim}.
The argument just given establishes the converse.
\yfim
\end{proof}

The next corollary exposes a strong relationship between perfectness of a test suite $T$ for a FSM $M$ and $p$-reduced FSMs that are $T$-alike to $M$.  
\begin{corollary}\label{cor:perfect}
Let $M$ be a $p$-reduced FSM and $T$ be a test suite for $M$. 
If $T$ is perfect for $M$ then any $p$-reduced $T$-alike FSM is isomorphic to $M$. 
\end{corollary}
\begin{proof}
Assume that $T$ is perfect for $M$ and  let $N$ be a $p$-reduced FSM that is $T$-alike $M$.
By Theorem~\ref{necessity-sufficiency-block}, we know that $N$ is bi-similar to $M$.
Then, $M$ and $N$ are isomorphic, using Theorem~\ref{completeness-isomorphim}.
\yfim
\end{proof}

\subsection{$p$-reduced Finite State Machines}

The converse of Corollary~\ref{cor:perfect} actually also holds.
But, since  Theorem~\ref{completeness-isomorphim} stipulates that \emph{all} $T$-alike FSMs must simulate the specification $M$, first we must show that any FSM can be $p$-reduced
without loosing the $T$-alikness property.

Recall from Lemma~\ref{lem:alike-trans} that $\yalike$ is an equivalence relation on $S$ on $M$.
We denote by $[s]$ the equivalence class of $s$ under the relation $\yalike$.
We now use the classical idea of taking quotients in order to construct a FSM 
$\ywb{M}$
that is $p$-reduced and alike to $M$.  Define 
$$\ywb{S}=\{[s]\yst s\in S, \text{and $s\ydio{\alpha}{\omega}$, some $\alpha\in\yias$,
$\omega\in\yoas$}\},$$ and
$\ywb{s_0}=[s_0]$.
Next, if $s\yalike r$ and $(s,x)\in D$, then Lemma~ \ref{lem:alike}(1) gives $(r,x)\in D$. 
We can then define $\ywb{D}=\big\{\big(\,[s],x\,\big)\big| (s,x)\in D\big\}$.
Since $([s],x)\in \ywb{D}$ implies $(s,x)\in D$, and Lemma~ \ref{lem:alike}(1), again,
would give $\delta(s,x)\yalike \delta(r,x)$ for all $r\in[s]$, we can define
$\ywb{\delta}\big([s],x\big)=\big[\delta(s,x)\big]$.
Finally, note that if $s\yalike r$ and $s\ydio{x}{a}p$, for some $p\in S$, $x\in\yia$ and $a\in\yoa$, then Lemma~ \ref{lem:alike}(1) gives $r\ydio{x}{a} q$, for some $q\in S$, that is,
$\lambda(s,x)=\lambda(r,x)$ whenever $s\yalike r$ and $x\in U(s)$.
Thus, we can define $\ywb{\lambda}\big([s],x\big)=\lambda(s,x)$. The construction of 
$\ywb{M}$ is complete.

\begin{defin}\label{M:p-reduced}
Let $M$ be a FSM.
Then $\ywb{M}=(\ywb{S},\ywb{s_0},\yia,\yoa,\ywb{\ydom},\ywb{\ydel},\ywb{\ylam})$ is the FSM
given by the preceding construction. 
\end{defin}

The foregoing construction satisfy a number of simple properties that will be useful later.
\begin{fact}\label{fact:m-to-equiv}
Let $s,r\in S$, and let $\alpha\in\yias$, $\omega\in\yoas$.
If $s\ydio{\alpha}{\omega}r$, then $[s]\ydio{\alpha}{\omega}[r]$.
\end{fact}
\begin{proof}
Assume that $s\ydio{x}{a}r$, with $x\in\yia$ and $a\in\yoa$. Then $\delta(s,x)=r$ and $\lambda(s,x)=a$.
From the construction of $\ywb{M}$ we get $\delta([s],x)=[r]$ and 
$\ywb{\lambda}([s],x)=a$.
Hence, $[s]\ydio{x}{a}[r]$, and the result follows by an easy induction
on $|\alpha|\geq 0$. \yfim
\end{proof} 

\begin{fact}\label{fact:equiv-to-m}
Let $r,q\in S$, and let $\alpha\in\yias$, $\omega\in\yoas$.
If $[r]\ydio{\alpha}{\omega}[q]$, then $r_1\ydio{\alpha}{\omega}q_1$, for
some $r_1, q_1\in S$ with $r\yalike r_1$ and
$q\yalike q_1$.
\end{fact}
\begin{proof}
Assume that $[r]\ydio{x}{a}[q]$, with $x\in\yia$ and $a\in\yoa$. Then $\ywb{\delta}([r],x)=[q]$ and $\ywb{\lambda}([r],x)=a$.
From $\ywb{\delta}([r],x)=[q]$, the construction of $\ywb{M}$ gives $r_1, q_1\in S$ with
$\delta(r_1,x)=q_1$, $r_1\yalike r$ and $q_1\yalike q$.
From $\ywb{\lambda}([r],x)=a$, we get $r_2\in S$ with $\lambda(r_2,x)=a$ and $r_2\yalike r$.
Hence, $r_1\yalike r_2$.

Since $r_1\ydio{x}{b} q_1$, this gives $r_2\ydio{x}{b} r_3$, for some $r_3\in S$.
But $\lambda(r_2,x)=a$, and so $a=b$ because machines are deterministic.
Collecting, we have $r_1\ydio{x}{a} q_1$, $r_1\yalike r$ and $q_1\yalike q$.
The result now follows using a simple induction on $|\alpha|$. \yfim
\end{proof} 

\begin{lemma}\label{lem:p-reduced}
Let $M$ be a FSM and $s,r\in S$.
Let $\ywb{M}$ be the FSM in Definition~\ref{def:p-reduced}.
If $[s]\neq [r]$, then $[s]\ynalike [r]$.
\end{lemma}
\begin{proof}
Assume $[s]\yalike [r]$ and show that $s\yalike r$.
First, we show that $U(s)\ysyd U(r)=\yemp$.
Let $\alpha\in U(s)$.
Then $s\ydio{\alpha}{\omega} p$, for some $p\in S$ and $\omega\in\yoas$.
Using Fact~\ref{fact:m-to-equiv}, we get  $[s]\ydio{\alpha}{\omega} [p]$.
Since $[s]\yalike [r]$, Lemma~\ref{lem:alike-trans} gives $[r]\ydio{\alpha}{\omega}[q]$, for some $[q]\in \ywb{D}$.
Using Fact~\ref{fact:equiv-to-m} we obtain $r_1\ydio{\alpha}{\omega} q_1$, for some $q_1 \in S$ with $r_1\yalike r$.
Hence, Lemma~\ref{lem:alike-trans} now gives $r\ydio{\alpha}{\omega} q_2$, for some $q_2\in S$. 
We conclude that $\alpha\in U(r)$, thus establishing that $U(s)\ysse U(r)$.
A similar argument gives $U(r)\ysse U(s)$, and so $U(s)=U(r)$, as needed.
To finish, let now $\alpha\in U(s)\cap U(r)$.
Then, $s\ydio{\alpha}{\omega} p$, for some $p\in S$. Repeating the preceding argument would
give, again, $r\ydio{\alpha}{\omega} r_2$, for some $r_2\in S$.
Hence, $\lambda(s,\alpha)=\omega=\lambda(r,\omega)$.
From Definition~\ref{disting-equiv-block} we conclude that $s\yalike r$. \yfim
\end{proof} 

At this point, we can already establish that $\ywb{M}$ is $p$-reduced.
\begin{corollary}\label{cor:m-p-rduced}
Let $\ywb{M}$ be the FSM in Definition~\ref{def:p-reduced}.
Then, $\ywb{M}$ is $p$-reduced.
\end{corollary} 
\begin{proof}
Let $[s]\in\ywb{S}$.
By construction, $s_0\ydio{\alpha}{\omega} s$, for some $\alpha\in\yias$, $\omega\in\yoas$.
Hence, Lemma~\ref{lem:alike}(2) gives $\ywb{s_0}\ydio{\alpha}{\omega} [s]$, because
$\ywb{s_0}=[s_0]$.
Further, if $[s]$ and $[r]$ are distinct, Lemma~\ref{lem:p-reduced} implies $[s]\ynalike [r]$.
\yfim
\end{proof}

In the next result, we use the same symbol, $\yalike$, to denote the alikeness relations between states of $M$, and also between states of $M$ and of $\ywb{M}$.
The context will always make clear which relation we are referring to.

\begin{lemma}\label{lem:alike-p-reduced}
Let $M$ be a FSM and $s,r\in S$.
Let $\ywb{M}$ be the FSM in Definition~\ref{def:p-reduced}.
If $s\yalike r$, then $s\yalike [r]$.
\end{lemma}
\begin{proof}
We first show that $U(s)\ysyd U([r])=\yemp$.
Let $\alpha\in U(s)$.
Since $s\yalike r$, Lemma~\ref{lem:alike}(2) gives $\alpha\in U(r)$.
Hence, using Fact~\ref{fact:m-to-equiv} we obtain $\alpha\in U([r])$, and so
$U(s)\ysse U([r])$.
Conversely, let $\alpha\in U([r])$.
Then, Fact~\ref{fact:equiv-to-m} gives $\alpha\in U(r_1)$, where $r_1\yalike r$.
Thus, $r_1\yalike s$, and so using Lemma~\ref{lem:alike}(2) we get $\alpha\in U(s)$.
This shows $U([r])\ysse U(s)$ and we may conclude that $U(s)=U([r])$.
Hence, $U(s)\ysyd U([r])=\yemp$ using Remark~\ref{rem:dif-sim}, as desired.

Now, let $\alpha\in U(s)\cap U([r])$.
Then, $s\ydio{\alpha}{\omega} s_1$, for some $s_1\in S$, $\omega\in \yoas$, and also
$[r]\ydio{\alpha}{\rho} [r_1]$, for some $[r_1]\in \ywb{S}$, $\rho\in\yoas$.
In order to get $\lambda(s,\alpha)=\ywb{\lambda}([r],\alpha)$ we just show that
$\omega=\rho$.
From $s\yalike r$, and using Lemma~\ref{lem:alike}(2), we have $r\ydio{\alpha}{\omega} r_2$,
for some $r_2\in S$ with $r_2\yalike s_1$.
Hence, by Fact~\ref{fact:m-to-equiv} we get $[r]\ydio{\alpha}{\omega} [r_2]$.
The determinism of $\ywb{M}$ now gives $\omega=\rho$. \yfim
\end{proof} 

We can now say that the $p$-reduction construction preserves alikeness.

\begin{corollary}\label{cor:m-alike-m}
Let $M$ be a FSM and let $\ywb{M}$ be the FSM in Definition~\ref{def:p-reduced}.
Then, $M\yalike \ywb{M}$.
\end{corollary} 
\begin{proof}
Since $s_0\yalike s_0$, Lemma~\ref{lem:alike-p-reduced} gives $s_0\yalike [s_0]$, and we
know that, by construction,  $\ywb{s_0}=[s_0]$.
\end{proof}

Besides preserving alikeness, the construction also yield bi-simulating machines. 

\begin{lemma}\label{lem:alike-bisim}
Let $M$ be a FSM and let $\ywb{M}$ be the FSM in Definition~\ref{def:p-reduced}.
Then, $M$ and $\ywb{M}$ are bi-similar.
\end{lemma}
\begin{proof}
Define the relation $R\ysse S\times \ywb{S}$ by letting $(s,[r])\in R$ iff $s\yalike r$.
Clearly, $(s_0,[s_0])\in R$.
Now, let $(s,[r])\in R$ with $s\ydio{x}{a} p$ for some $p\in S$, $x\in\yia$, $a\in\yoa$.
Since $s\yalike r$, Lemma~\ref{lem:alike}(1) gives $r\ydio{x}{a} q$ for some $q\in S$ with 
$q\yalike p$.
Then Fact~\ref{fact:m-to-equiv} gives $[r]\ydio{x}{a}[q]$.
But $(p,[q])\in R$, and we conclude that $R$ is a simulation relation.
For the other direction, define the raletion $L\ysse \ywb{S}\times S$ where 
$([r],s)\in L$ iff $r\yalike s$.
Again $([s_0],s_o)\in L$ clearly holds.
Let $([s],r)\in L$ with $[s]\ydio{x}{a} [q]$ for some $[q]\in\ywb{S}$, $a\in\yoa$, $x\in\yia$.
By Fact~\ref{fact:equiv-to-m}, we get $s_1\ydio{x}{a} q_1$ for some $s_1, q_1\in S$ with
$s\yalike s_1$ and $q\yalike q_1$. 
Since $([r],s)\in L$, we have $s\yalike r$, and so $r\yalike s_1$.
From $s_1\ydio{x}{a} q_1$ we conclude that $r\ydio{x}{a} q_2$, for some $q_2\in S$ with
$q_2\yalike q_1$, using Lemma~\ref{lem:alike}(1).
Thus, $q_2\yalike q$, and so $([q],q_2)\in L$, and we conclude that $L$ is also a simulation relation.  \yfim
\end{proof} 

The desired converse to Corollary~\ref{cor:perfect} can now be established.  

\begin{corollary}\label{cor:reduced-to-perfect}
Let $M$ be a $p$-reduced FSM and let $T$ be a test suite for $M$.
Assume that all $p$-reduced $T$-alike FSMs are isomorphic to $M$.
Then $T$ is perfect for $M$.
\end{corollary} 
\begin{proof}
In view of Theorem~\ref{necessity-sufficiency-block}, it suffices to show that any FSM that is $T$-alike to $M$ is also bi-similar to $M$.
Let $N$ be $T$-alike to $M$.
Let $\ywb{N}$ be as in Definition~\ref{def:p-reduced}.
By Corollary~\ref{cor:m-p-rduced} $N$ is $p$-reduced, and by Corollary~\ref{cor:m-alike-m} we have $N\yalike \ywb{N}$.
Now, in view of Remark~\ref{rema:equal-Us}(2) we conclude that $N\yalike_T \ywb{N}$.
Since we already have $M\yalike N$, using Lemma~\ref{lem:alike-trans} and Remark~\ref{rema:alike-trans}, we conclude that $M\yalike \ywb{N}$. 
So, $\ywb{N}$ is $p$-reduced and $T$-alike $M$.
By the hypothesis we know that $M$ and $\ywb{N}$ are isomorphic.
Hence, using Theorem~\ref{completeness-isomorphim}, we know that $M$ and $\ywb{N}$ are
bi-similar.
But $\ywb{N}$ and $N$ are also bi-similar, using Lemma~\ref{lem:alike-bisim}.
Finally, using Fact~\ref{fact:sim-trans}, we conclude that $M$ and $N$ are bi-similar,
as desied. \yfim
\end{proof}

We can now collect the results of this section in the following theorem.
\begin{theorem}
Let $M$ be a $p$-reduced FSM and let $T$ be a test suite for $M$.
Then $T$ is perfect for $M$ iff all $p$-reduced $T$-alike FSMs are isomorphic to $M$.
\end{theorem}
\begin{proof}
Use Corollaries~\ref{cor:perfect} and \ref{cor:reduced-to-perfect}. \yfim
\end{proof}


\section{Completeness and Perfectness}\label{completenessxperfectness}

In this section we investigate the relationship between completeness and perfectness.
We show that a test suite $T$ that is not $n$-complete for a FSM $M$ can not also be perfect for $M$, for any $n\geq 1$.
In the other direction, we also show that there are test suites $T$ which are perfect for $M$,
but not $n$-complete for $M$, for $n\geq 2$.

We start by showing that perfectness only holds when $n$-completeness also holds. 
Let $M$ be a FSM and let $T$ be a test suite for $M$. 
We want to prove that if $T$ is not $n$-complete for $M$, then $T$ is not perfect for $M$, where $n\geq 1$.
This will show that perfectness is at least as strong a condition as is completeness.

First, we need a measure on the length of blocking test cases in a test suite.
Let $\alpha\in \yias$ be an input string for $M$.
Define $F(M,\alpha)$ as:
$$F(M,\alpha)=\max\big\{ |\beta|\,:\, \alpha=\beta x\gamma, \text{with $\beta\in U(s_0)$, $\beta x\not\in U(s_0)$, $x\in\yia$}\big\}.$$
That is, $F(M,\alpha)$ is the maximum length of a prefix of $\alpha$ which does not block in $M$. 
For a test suite $T\ysse\yias$ we overload the notation and define 
$F(M,T)=\sum\limits_{\alpha\in T} F(M,\alpha)$.
\begin{fact}\label{fact:limite}
Given  a FSM $M$ and a test suite $T$ for $M$, we have the upper bound
$F(M,T)\leq \sum\limits_{\alpha\in T}|\alpha|$.
\end{fact}
\begin{proof}
Immediate.
\end{proof}

Now, fix a FSM $M$, a test suite $T$, and assume that $T$ is not $n$-complete for $M$, for some $n\geq 1$.
Then, there is a FSM $N$ such that $M\ynequ N$ and $M \yequ_T N$.
So, we have some $\sigma=x_1x_2\ldots x_{n+1}$, where $n\geq 0$ and $x_i\in\yia$ ($1\leq i\leq n+1$), and such that 
\begin{equation}\label{eq:sig}
\sigma\not\in T \qquad\text{and\qquad $\sigma\in U(s_0)$}.
\end{equation}
Let 
\begin{equation}\label{eq:cinco}
s_0 \ydio{x_1}{a_1} s_1 \ydio{x_2}{a_2} s_2 \cdots s_{n-1} \ydio{x_n}{a_n} s_n
\ydio{x_{n+1}}{a_{n+1}} s_{n+1}.
\end{equation}

We show how to construct a sequence of FSMs $N_i$  that satisfy, for all $i\geq 0$:
\begin{enumerate}
\item $N_i$ is a tree rooted at $q_0$.
\item $\sigma\in U_i(q_0)$.
\item for all $\alpha\in U_i(q_0)\cap T$ we have:
\begin{enumerate}
\item $\alpha\in U(s_0)$.
\item If $q_0 \yddio{\alpha}{\omega}{N_i}$ and $s_0 \yddio{\alpha}{\eta}{M}$, then $\omega=\eta$. 
\end{enumerate}
\end{enumerate}
In order to ease the notation, we denote the states in each $N_i$ as $q_0$,
$q_1$, $q_2$, \ldots, with $q_0$ the initial state.
Moreover, by $U_i(q_0)$ we mean the set of all input strings $\alpha$ such that
$q_0 \yddio{\alpha}{\omega}{N_i}$, for some output string $\omega$.

We start by defining $N_0$ as the FSM containing the transitions:
\begin{equation}\label{eq:um}
q_0 \ydio{x_1}{a_1} q_1 \ydio{x_2}{a_2} q_2 \cdots s_{n-1} \ydio{x_n}{a_n} q_n
\ydio{x_{n+1}}{b} q_{n+1},
\end{equation}
where $b\ne a_{n+1}$.
It is clear that $N_0$ is a tree rooted at $q_0$, and that $\sigma\in U_0(q_0)$, and so 
properties (1) and (2) hold for $N_0$.
Now, let $\alpha\in U_0(q_0)\cap T$.
Since $\sigma\not\in\ T$, we conclude that $\alpha$ is a prefix of $x_1x_2\cdots x_n$, and so 
property (3) also holds for $N_0$.

Now assume that $N_i$ has been constructed satisfying properties (1)--(3), for some $i\geq 0$.
If there is some input string $\alpha \in U(s_0)\cap T$ such that $\alpha\not\in U_i(q_0)$
we show how to construct $N_{i+1}$.
Since   $\alpha\not\in U_i(q_0)$, we can write $\alpha=y_1y_2\cdots y_k x \beta$, where
$k\geq 0$, $y_j\in \yia$ ($1\leq j\leq k$), $x\in\yia$, and where we also have
$y_1y_2\cdots y_k \in U_i(q_0)$, $y_1y_2\cdots y_k x\not\in U_i(q_0)$.
So, in $N_i$ we have the transitions
\begin{equation}\label{eq:dois}
r_0 \ydio{y_1}{b_1} r_1 \ydio{y_2}{b_2} r_2 \cdots r_{k-1} \ydio{y_k}{b_k} r_k
\end{equation}
with $r_0=q_0$ and with no transition out of $r_k$ on input $x$.
Since $\alpha\in U(s_0)$, in $M$ we get 
\begin{equation}\label{eq:tres}
p_0 \ydio{y_1}{b_1} p_1 \ydio{y_2}{b_2} p_2 \cdots p_{k-1} \ydio{y_k}{b_k} p_k
 \ydio{x}{c} p_{k+1},
\end{equation} 
for some $c\in \yia$ and with $p_0=s_0$.
We define $N_{i+1}$ from $N_i$ by adding to it a transition  $r_k \ydio{x}{c} r$, and where
$r$ is a new state not present in $N_i$.

Since $N_i$ is a tree rooted at $q_0$, then so is $N_{i+1}$ because $r$ is a new state.
Then property (1) holds for $N_{i+1}$. 
Also, since all transitions from $N_i$ are present in $N_{i+1}$, then property (2), trivially, also holds for $N_{i+1}$.

Now, let $\gamma\in U_{i+1}(q_0)\cap T$.
Since $\gamma\in U_{i+1}(q_0)$ we have two cases:
\begin{itemize}
\item {\sc Case 1}: the new transition $r_k \ydio{x}{c} r$ does not occur in $\gamma$.
Then, clearly, $\gamma\in U_{i}(q_0)$,  and so (3a) and (3b) hold because $N_i$ satisfies 
property (3). 

\item {\sc Case 2}: the new transition $r_k \ydio{x}{c} r$ occurs in $\gamma$.
Since $r$ is a new state, we can write $\gamma=\delta x$, where $\delta\in U_i(q_0)$ and
$q_0\yddio{\delta}{\eta}{N_{i+1}} r_k\yddio{x}{c}{N_{i+1}} r$.
Since $N_i$ is a tree rooted at $q_0$, there is only one path from $q_0$ to $r_k$.
Hence, from Eq. (\ref{eq:dois}) we get $\delta=y_1y_2\cdots y_k$, and $\eta = b_1b_2\cdots b_k$. 
From Eq. (\ref{eq:tres}) we get $s_0\yddio{\delta}{\eta}{M} p_k\yddio{x}{c}{M} p_{k+1}$, 
and property (3) holds for $N_{i+1}$.
\end{itemize}
We conclude that properties (1)--(3) hold for $N_{i+1}$, as desired.

Because  $\alpha=y_1y_2\cdots y_k x \beta$, $y_1y_2\cdots y_k x\not\in U_i(q_0)$ and the
construction of $N_{i+1}$ gives $y_1y_2\cdots y_k x\in U_{i+1}(q_0)$ we conclude 
that $F(N_i,\alpha) < F(N_{i+1},\alpha)$.
Since we also have $\alpha\in T$, we then get  $F(N_i,T) < F(N_{i+1},T)$.

The preceding discussion shows that we can construct the sequence of FSMs $N_0$, $N_1$, \ldots
satisfying properties (1)--(3), and with  $F(N_i,T) < F(N_{i+1},T)$,  as long as we 
have input strings $\alpha_i\in U(s_0)\cap T$ such that $\alpha_i\not\in U_i(q_0)$, $i\geq 0$.

\begin{fact}\label{fact:theend}
There is some $\ell \geq 0$ such that  there is no $\alpha\in U(s_0)\cap T$ and such that
$\alpha\not\in U_\ell(q_0)$.
\end{fact}
\begin{proof}
Fact \ref{fact:limite} establishes an upper limit to the sequence $F(N_0,T) < F(N_1,T)< \cdots$. \yfim
\end{proof}

Now we can take the test case $\sigma$, that is not in $T$, and use the fact that the construction
gives $\sigma\in U(q_\ell)$ to show that $T$ is not, in fact, perfect for $M$.

From Eqs. (\ref{eq:sig}) and (\ref{eq:cinco}) we can write $s_0\yddio{\sigma}{\omega a_{n+1}}{M}$, where
$\omega=a_1a_2\cdots a_n$.
From Eq. (\ref{eq:um}) and property (2), we get $s_0\yddio{\sigma}{\omega b}{N_\ell}$.
Since $a_{n+1}\neq b$ we conclude that $M\ynalike N_\ell$.
If $T$ was perfect for $M$ we would have $M\ynalike_T N_\ell$.
We now show that this leads to contradictions.
There are two cases:
\begin{itemize}
\item {\sc Case A}:
there is some input string $\alpha\in U(s_0)\cap U_\ell(q_0)\cap T$ such that 
$s_0\yddio{\alpha}{\omega_1}{M}$, $q_0\yddio{\alpha}{\omega_2}{N_\ell}$, and $\omega_1\neq \omega_2$.
This contradicts property (3b). 
\item {\sc Case B}: there is some input string $\alpha\in (U(s_0)\ysyd U_\ell(q_0))\cap T$.
If $\alpha\in U_\ell(q_0)\cap T$ and $\alpha\not\in U(s_0)$, we contradict property (3a).
If $\alpha\in U(s_0)\cap T$ and $\alpha\not\in U_\ell(q_0)$, we contradict Fact \ref{fact:theend}.
\end{itemize}
We conclude that $T$ is not perfect for $M$.

\begin{fact}\label{fact:atlast}
Let $M$ be a FSM, and let $T$ be a test suite that is not $n$-complete for $M$, for some $n\geq 1$.
Then, $T$ is not perfect for $M$.
\end{fact}
\begin{proof}
From the preceding discussion. \yfim
\end{proof}

Next we also show that when $T$ is $n$-complete for $M$, $n\geq 1$, it may be the case that $T$ is not perfect for $M$.
Let the input and output alphabets be $\yia=\yoa=\{0,1\}$, and let 
$M$ be the specification with $n$ states given by the transitions 
$s_i \ydio{0}{0} s_{i+1}$, $0\leq i < n$.
Let $T=\{0^n,0^{n-1}\}$ be a test suite for $M$. 
We argue that $T$ is $n$-complete for $M$.
From Definitions~\ref{disting-equiv-partial} and~\ref{def:complete-partial}, if that were not the case, we would have a FSM $N$ with $U(s_0)\ysse U(q_0)$, and such that $M\ynequ N$ and $M\yequ_T N$.
Since $U(s_0)\ysse U(q_0)$ and $U(s_0)=\{0^{n-1}\}$, we get  
$U(s_0)\cap U(q_0)\cap T=\{0^{n-1}\}$.
Hence $M\yequ_T N$ gives $\lambda(s_0,0^{n-1})=0^{n-1}=\mu(q_0,0^{n-1})$.
Since we also have $U(s_0)\cap U(q_0)\cap \yias=\{0^{n-1}\}$, Definition~\ref{disting-equiv-partial} and $M\ynequ N$ would require  $\lambda(s_0,\alpha)\neq\mu(q_0,\alpha)$ for some 
$\alpha\in \{0^{n-1}\}$, and we reached a contradiction.

We now argue that $T=\{0^n,0^{n-1}\}$ is not perfect for the same specification $M$. 
Let $N$ be the FSM with the transitions $q_i \ydio{0}{0} q_{i+1}$ for $0\leq i<n$, and 
also $q_{n-1} \ydio{1}{1} q_{n-1} $.
It is clear that $0^{n-1}1\in (U(s_0)\ysyd U(q_0))\cap \yias$.
Hence, from Definition~\ref{disting-equiv-block}, we see that $M\ynalike N$.
Since $T=\{0^n,0^{n-1}\}$, it is clear that $(U(s_0)\ysyd U(q_0))\cap T=\yemp$.
Moreover, $U(s_0)\cap U(q_0)\cap T=\{0^{n-1}\}$, and so 
$\lambda(s_0,\alpha)=\mu(q_0,\alpha)$ for all $\alpha\in U(s_0)\cap U(q_0)\cap T$.
From Definition~\ref{disting-equiv-block} we get $M\yalike_T N$.
Hence, Definition~\ref{def:n-complete-relax} says that $T$ is not perfect for $M$.

\begin{corollary}\label{cor:perfec-complete}
Let $M$ be a FSM.
Then the following holds:
\begin{enumerate}
\item If $T$ is a test suite which is perfect for $M$, then $T$ is also $n$-complete for $M$, for all $n\geq 1$.
\item For all $n\geq 1$ there are test suites which are $n$-complete but not perfect for $M$.
\end{enumerate}
\end{corollary}
\begin{proof}
From the preceding discussion. \yfim
\end{proof}

\section{Test Suite Completeness and the Size of Implementations}\label{bounding-impls}

In this section we show that if one allows for too large implementations, then test completeness, in the classical sense, is lost.
More specifically, if $T$ is a test suite for a FSM $M$, then $T$ is not $n$-complete for $M$, where $n>k|S|$ is the number of states in implementation machines, and $k$ is a constant that depends only on $T$. 
This means that $T$ may not be able to detect all faults in implementations with $n$ or more states.
In the sequel, we use this result to also establish a bound on the size of implementation models when testing in the presence of blocking test cases, \emph{i.e.}, when
testing for perfectness. 

First, we establish some notation.
Let $\sigma=x_0x_1\cdots x_k$ be a sequence of symbols over an alphabet.
Then $\sigma_{i,j}$ ($0\leq i<j\leq k+1$) indicates the substring $x_ix_{i+1}\cdots x_{j-1}$. 
Let $\alpha$ be another sequence of symbols over the same alphabet.
We say that $\sigma$ is \emph{embedded} in $\alpha$ if and only if
there are sequences of symbols $\beta_i$ ($0\leq i\leq k+1$) such that 
$\alpha=\beta_0x_0\beta_1x_1\cdots \beta_kx_k \beta_{k+1}$.
Let $T$ be a test suite for a FSM $M$ and let $\sigma\in T$.
We say that $\sigma$ is \emph{extensible} in $T$ if and only if 
$\sigma=\sigma_1\sigma_2$ and there is some non-null $\gamma$ such that $\sigma_1\gamma\sigma_2$ is in $T$. Otherwise, $\sigma$ is \emph{non-extensible} in $T$.

From this point on, we fix a reduced FSM $M$ and a test suite $T$ for $M$.
Also, we fix $\sigma=x_0x_1\cdots x_k$, $k\geq 0$, as a smallest non-extensible
test case in $T$.
Trivially, such a test case always exists.
The following construction, and the series of accompanying facts, will give us the
desired result about the size of implementations when testing for completeness by.

\begin{remark}\label{rem:sem-testes}
If $T\cap U(s_0)=\yemp$ then any FSM is trivially $T$-equivalent to $M$.
Moreover, if $\sigma=\yeps$, then $T=\{\yeps\}$ and, again, 
any FSM is trivially $T$-equivalent to $M$.
Since $M$ is reduced, one can easily construct a one-state FSM that is not
equivalent to $M$. 
Hence, in both cases, $T$ would not be $1$-complete for $M$.
We, therefore, can assume that such a non-null $\sigma\in T\cap U(s_0)$. 
\end{remark} 
Since $\sigma\in U(s_0)$, we get transitions $\pi_i: s_i\ydio{x_i}{a_i}{s_{i+1}}$ in $M$
($0\leq i < k$).
Those are the \emph{distinguished transitions of $M$}.
Moreover, since $M$ is reduced, using Remark~\ref{rem:reduced} we have $s_{k+1}\ydio{z}{a} s'$ in $M$, for some
$z\in\yia$, $a\in \yoa$ and $s'\in S$.
We call this the \emph{marked transition of $M$}.
 
We now construct a FSM $N$ using the same input and output alphabets, respectively $\yia$ and $\yoa$, of $M$.
A simple example illustrating the construction is presented right after 
Theorem~\ref{theo:bound-completeness}.
Let $Q=S\times [0,k+1]$, that is, the states of $N$ are pairs $[q,i]$ where 
$q$ is a state of $M$ and $0\leq i\leq k+1$.
The initial state of $N$ is $q_0=[s_0,0]$.
We complete the specification of $N$ by listing its transitions:
\vspace*{-2ex}
\begin{list}{}{\setlength{\leftmargin}{3ex}\setlength{\labelwidth}{1ex}}
\item[(a)] If $s\ydio{y}{b} r$ is not a distinguished transition of $M$, let
$[s, i]\ydio{y}{b} [r,i]$ be a transition in $N$, for all $i$, $0\leq i\leq k$.
\item[(b)] For all distinguished transitions $s_i\ydio{x_i}{a_i}{s_{i+1}}$ of $M$,
let $[s_i, i]\ydio{x_i}{a_i} [s_{i+1},i +1]$ be a transition in $N$.
We call these the \emph{distinguished transitions of $N$}.
\item[(c)] If $s\ydio{y}{b} r$ is not the marked transition of $M$, we let
$[s,k+1]\ydio{y}{b} [r,k+1]$ be a transition in $N$. 
\item[(d)] For the marked transition of $M$, $s_{k+1}\ydio{z}{a} s'$, we let
$[s_{k+1},k+1]\ydio{z}{b} [s',k+1]$, for some $b\neq a$, be a transition in $N$. 
\end{list}
This completes the specification of $N$.
Easily, $N$ has $(|\sigma|+1)|S|$ states.

The next facts make explicit the behavior of the construction.
\begin{fact}\label{fact:M-to-N}
Let $\pi:s\ydio{\alpha}{\omega} p$ in $M$ and take $0\leq i\leq k+1$.
Then in $N$ we must have $[s,i]\ydio{\alpha}{\omega'}[p,j]$ for some $j\geq i$.
Moreover, $\omega=\omega'$ if the marked transition of $M$ does not occur in $\pi$.
\end{fact}
\begin{proof}
By induction on $|\alpha|=n\geq 0$.
When $n=0$ the result follows immediately.

For the induction step, let $\alpha=\beta x$, $\omega=\rho a$, with $x\in \yia$, $a\in \yoa$, and 
$\pi:s\ydio{\beta}{\rho}r\ydio{x}{a} p$.
The induction hypothesis gives $\pi_1:[s,i]\ydio{\beta}{\rho'}[r,j]$ in $N$, with $j\geq i$.

If $j=k+1$, then items (c) and (d) in the construction of $N$ give
$[r,j]\ydio{x}{a'}[p,j]$ in $N$.
Then, clearly, $[s,i]\ydio{\alpha}{\omega'}[p,j]$ in $N$, where $\omega'=\rho' a'$.
Moreover, if the marked transition of $M$ does not occur in $\pi$ then
the induction hypothesis gives $\rho=\rho'$.
Also,  since $r\ydio{x}{a} p$ is not the marked transition of $M$, item (c) of the construction of $N$ yields $a'=a$.
We conclude that $\omega=\rho a=\rho'a'=\omega'$, as desired.

Now take $j<k+1$.
Then items (a) and (b) of the construction give $[r,j]\ydio{x}{a'}[p,\ell]$ in $N$
where $\ell=j$ or $\ell=j+1$. 
Hence, $[s,i]\ydio{\alpha}{\omega'}[p,j]$ with $\omega'=\rho' a'$ and, in any case, $\ell\geq j\geq i$, as desired.
Again, if the marked transition of $M$ does not occur in $\alpha$ then 
we get $\rho=\rho'$ using the induction hypothesis. Clearly, from items (a) and (b)
we have $a'=a$. This readily gives   $\omega=\rho a=\rho'a'=\omega'$, concluding the 
proof.\yfim
\end{proof}

The next result gives the converse.
\begin{fact}\label{fact:N-to-M}
Let $\pi:[s,i]\ydio{\alpha}{\omega} [p,j]$ in $N$.
Then we have: (i) $j\geq i$, (ii) $\sigma_{i,j}$ is embedded in $\alpha$,
 and (iii) $s\ydio{\alpha}{\omega'} p$ in $M$.
Moreover, $\omega=\omega'$ if the marked transition of $N$ does not occur in $\pi$.
\end{fact}
\begin{proof}
By induction on $|\alpha|=n\geq 0$.
When $n=0$ the result follows easily.

For the induction step, let $\alpha=\beta x$, $\omega=\rho a$, with $x\in \yia$, $a\in \yoa$, and 
$\pi':[s,i]\ydio{\beta}{\rho}[r,\ell]\ydio{x}{a} [p,j]$.
The induction hypothesis gives $\ell\geq i$, $\sigma_{i,\ell}$ embedded in $\beta$,
and $s\ydio{\beta}{\rho'} r$ in $M$.
Following the items in the construction of $N$ we have four cases for the transition 
$[r,\ell]\ydio{x}{a} [p,j]$:
\vspace*{-2ex}
\begin{list}{}{\setlength{\leftmargin}{3ex}\setlength{\labelwidth}{1ex}}
\item[(a)] It was added because of item (a).
Then, $\ell=j$ and $r\ydio{x}{a} p$ is in $M$.
We get $j=\ell\geq i$ and $\sigma_{i,j}=\sigma_{i,\ell}$ is embedded in $\alpha$, as desired.
Composing we get $s\ydio{\beta x}{\omega'} p$ in $M$, with $\beta x=\alpha$ and
$\rho'a=\omega'$.
If the marked transition of $M$ does not occur in $\pi$, then
 $\rho=\rho'$ by the induction hypothesis. So, $\omega=\rho a=\rho'a=\omega'$,
as we wanted.
\item[(b)] It was added because of item (b).
Then, $x=x_\ell$, $j=\ell+1$, and $r\ydio{x}{a} p$ in $M$.
Clearly, (i) and (iii) hold, with $\omega'=\rho'a$.
Also, $\sigma_{i,j}=\sigma_{i,\ell+1}=\sigma_{i,\ell}x_\ell$.
Since $\alpha=\beta x=\beta x_\ell$ and $\sigma_{i,\ell}$ is embedded in $\beta$,
we conclude that $\sigma_{i,j}$ is embedded in $\alpha$.
If the marked transition of $M$ does not occur in $\pi$, then we proceed as in case (a),
and obtain $\omega=\rho a=\rho'a=\omega'$, as needed.
\item[(c)] It was added because of item (c).
Now we have $\ell=k+1=j$ and $r\ydio{x}{a}$ in $M$, showing that (i) and (iii) hold
with $s\ydio{\beta x}{\omega'} p$ and $\omega'=\rho'a$.
We have that $\sigma_{i,\ell}=\sigma_{i,j}$ is already embedded in $\beta$ and so
its also embedded in $\alpha$, given that $\alpha=\beta x$.
The reasoning to obtain $\omega=\omega'$ is the same as in case (a).
\item[(d)] It was added because of item (d).
Proceed exactly as in case (c).
Now, the marked transition of $N$ does 
occur in $\pi$ and so the last statement of the Fact holds vacuously.
This last case concludes the proof.\yfim
\end{list}
\end{proof}

The last two results already establish that the same sequences of input symbols
will run in both machines.
\begin{fact}\label{fact:equalUs}
$U(s_0)=U(q_0)$.
\end{fact}
\begin{proof}
Recall that $q_0=[s_0,0]$.
Let $s_0\ydi{\alpha}$ in $M$.
Using Fact~\ref{fact:M-to-N} we get $[s_0,0]\ydi{\alpha}$ in $N$.
Hence, $U(s_0)\ysse U(q_0)$.
In a similar way we can get $U(q_0)\ysse U(s_0)$ using Fact~\ref{fact:N-to-M},
and the result follows.\yfim
\end{proof}

We are now in a position to show that $M$ and $N$ are $T$-equivalent.
\begin{fact}\label{fact:M-equiv-N}
$M\yequ_T N$.
\end{fact}
\begin{proof}
We go by contradiction.
Assume we have $\alpha x\in T\cap U(s_0)\cap U(q_0)$, $x\in\yia$ such that
$s_0\ydio{\alpha}{\omega} s\ydio{x}{a} r$ in $M$ and 
$[s_0,0]\ydio{\alpha}{\omega} [q,i]\ydio{x}{b} [p,j]$ in $N$, with $a\neq b$. 
Fact~\ref{fact:N-to-M} gives $s_0\ydi{\alpha} q$ in $M$.
But we already have $s_0\ydi{\alpha} s$ in $M$, and so we conclude that $s=q$.
Using Fact~\ref{fact:N-to-M} again, from $s\ydi{x} r$ in $M$ and $[s,i]\ydi{x}[p,j]$ in $N$
we get $p=r$.
We can now write $\pi:[s,i]\ydio{x}{b} [r,j]$ in $N$ and $s\ydio{x}{a} r$ in $M$ with $a\neq b$.
From the construction of $N$ we conclude that $\pi$ is the marked transition of $N$.
Hence, $i=j=k+1$.
We now have $[s_0,0]\ydio{\alpha}{\omega}[s,k+1]$ in $N$.
From Fact~\ref{fact:N-to-M}, $\sigma=\sigma_{0,k+1}$ is embedded in $\alpha$ and so
$\sigma$ is embedded in $\alpha x$.
Since $\alpha x\in T$, we conclude that $\sigma$ is extensible in $T$.
But this contradicts the choice of $\sigma$, completing the proof.\yfim
\end{proof}

In the opposite direction, the next result shows that $M$ and $N$ are not equivalent.
\begin{fact}\label{fact:M-not-equiv-N}
$M\ynequ N$.
\end{fact}
\begin{proof}
Since $\sigma\in U(s_0)$, Fact~\ref{fact:equalUs} gives $\sigma\in U(q_0)$.
By the choice of $\sigma$, in $M$ we have $s_0\ydio{\sigma}{\omega} s_{k+1}$.
Further, by the choice of $z$ and $a$, we have $s_{k+1}\ydio{z}{a} s'$ in $M$.
Hence, $s_0\ydio{\sigma z}{\omega a}s'$ in $M$.
Item (b) of the construction of $N$ gives $[s_i,i]\ydio{x_i}{a_i}[s_{i+1},i+1]$, $0\leq i\leq k$.
Then, $[s_0,0]\ydio{\sigma}{\omega}[s_{k+1},k+1]$ in $N$.
By item (d) of the construction of $N$ we get 
$[s_{k+1},k+1]\ydio{z}{b} [s',k+1]$ in $N$.
Composing, we obtain $[s_0,0]\ydio{\sigma z}{\omega b} [s',k+1]$ in $N$.
This shows that $M\ynequ N$, because $a\neq b$.\yfim
\end{proof}

Collecting, we can show that a test suite $T$ will not be $n$-complete for
a FSM $M$ when
$n$ is larger than a certain bound, which depends only on $M$ and $T$.
\begin{theorem}\label{theo:bound-completeness}
Let $M$ be a FSM and let $T$ be a test suite for $M$.
Let $\sigma$ be a shortest test case in $T$ that is non-extensible in $T$.
Then $T$ is not $\big((|\sigma|+1)|S|\big)$-complete for $M$.
\end{theorem}
\begin{proof}
The construction of $N$ yields a machine that is $T$-equivalent to $M$, 
using Fact~\ref{fact:M-equiv-N}.
We also know that $M$ and $N$ are not equivalent, by Fact~\ref{fact:M-not-equiv-N}.
Also, using Fact~\ref{fact:equalUs}, we know that $U(s_0)\ysse U(q_0)$.
Since $N$ has $n=(|\sigma|+1)\times |S|$ states, 
Definition~\ref{def:complete-partial} says that $T$ is not $n$-complete for $M$.
\yfim
\end{proof}

Next, we give a simple example to illustrate the construction of machine $N$. 
Let $\ymfm$ be a specification FSM as depicted in Figure~\ref{fsm-spec}. 
The set of states is $S=\{s_0,s_1\}$, $\yia=\yoa=\{0,1\}$, and $D,\delta,\lambda$ are given as depicted in the figure. 
Note that $M$ is a partial FSM since $(s_1,1)\notin D$. 
Also let $T=\{0000,100\}$ be a test suite for $M$.
We notice that $T$ is $2$-complete for $M$, \emph{i.e.}, for implementation FSMs with at most as many states as $M$. 
This can be checked by using the algorithm described in~\cite{BonifacioSEFM2014,TR-IC-13-33}. 

Now take $\sigma=100$ as the shortest test case in $T$ that is non-extensible in $T$. 
We apply items (a) to (d) of the construction of $N$, thus obtaining a machine with $(|\sigma|+1)|S|=(3+1)2=8$ states. 
From item (a) we create transitions $[s_0,i]\ydio{0}{1} [s_0,i]$, for all $i$, $0\leq i\leq 2$. 
We also obtain the distinguished transitions $[s_0, 0]\ydio{1}{1} [s_1,1]$, $[s_1, 1]\ydio{0}{0} [s_1,2]$, $[s_1, 2]\ydio{0}{0} [s_1,3]$ $[s_0, 1]\ydio{1}{1} [s_1,2]$, $[s_0, 2]\ydio{1}{1} [s_1,3]$ and $[s_1, 0]\ydio{0}{0} [s_1,1]$ from item (b). 
From item (c) we get the transitions $[s_0,3]\ydio{0}{1} [s_0,3]$, $[s_0,3]\ydio{0}{1} [s_0,3]$ and $[s_0,3]\ydio{1}{1} [s_1,3]$. 
Finally we complete machine $N$ with the marked transition $[s_3,3]\ydio{0}{1} [s_3,3]$ as required by item (d). 
Machine $N$ is depicted in Figure~\ref{fig:fsm-counter-all}. 
It is a simple matter to see that states $[s_0,1]$, $[s_0,2]$, $[s_0,3]$ and $[s_1,0]$ are not reachable in $N$. 
Then we can remove them in order to obtain a reduced FSM as depicted in Figure~\ref{fig:fsm-counter}.  
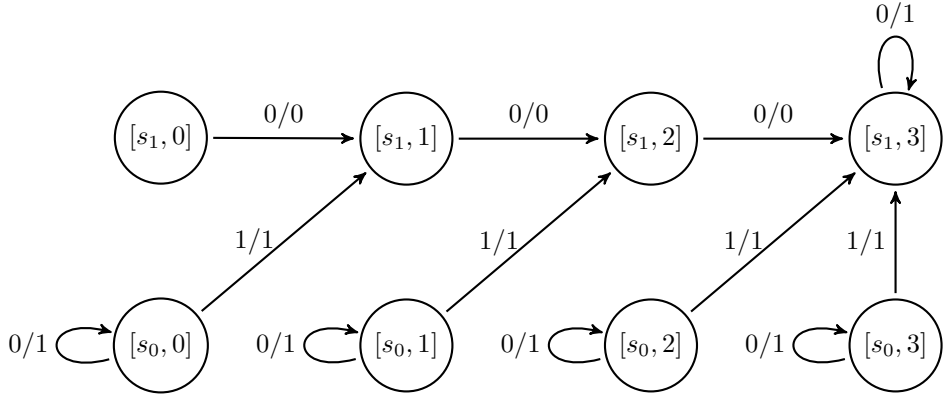
\begin{figure}[htb]
\center

\begin{tikzpicture}[node distance=1cm, auto]
  \node[ punkt] (q0) {$[s_0,0]$};
  \node[punkt, inner sep=3pt,right=2cm of q0] (q1) {$[s_0,1]$};
  \node[punkt, inner sep=3pt,right=2cm of q1] (q2) {$[s_0,2]$};
  \node[punkt, inner sep=3pt,right=2cm of q2] (q3) {$[s_0,3]$};
  \node[punkt, inner sep=3pt,above=1.5cm of q0] (q4) {$[s_1,0]$};
  \node[punkt, inner sep=3pt,above=1.5cm of q1] (q5) {$[s_1,1]$};
  \node[punkt, inner sep=3pt,above=1.5cm of q2] (q6) {$[s_1,2]$};
  \node[punkt, inner sep=3pt,above=1.5cm of q3] (q7) {$[s_1,3]$};

  \path (q0)   edge[pil, loop left]
                	node[anchor=left,left]{$0/1$} (q0);
  \path (q1)   edge[pil, loop left]
                	node[anchor=left,left]{$0/1$} (q1);
  \path (q2)   edge[pil, loop left]
                	node[anchor=left,left]{$0/1$} (q2);
  \path (q3)   edge[pil, loop left]
                	node[anchor=left,left]{$0/1$} (q3);

\path (q0)    edge [ pil]
                	node[anchor=north,left]{$1/1$} (q5);
\path (q1)    edge [ pil]
                	node[anchor=north,left]{$1/1$} (q6);
\path (q2)    edge [ pil]
                	node[anchor=north,left]{$1/1$} (q7);
\path (q3)    edge [ pil]
                	node[anchor=north,left]{$1/1$} (q7);

\path (q4)    edge [ pil]
                	node[anchor=north,above]{$0/0$} (q5);
\path (q5)    edge [ pil]
                	node[anchor=north,above]{$0/0$} (q6);
\path (q6)    edge [ pil]
                	node[anchor=north,above]{$0/0$} (q7);
  \path (q7)   edge[pil, loop above]
                	node[anchor=north,above]{$0/1$} (q7);
\end{tikzpicture}
\caption{A candidate implementation $N$.}
\label{fig:fsm-counter-all}
\end{figure}
Note that we have renamed states as $q_0=[s_0,0]$, $q_1=[s_1,1]$, $q_2=[s_1,2]$, and $q_3=[s_1,3]$. 
\begin{figure}[htb]
\center

\begin{tikzpicture}[node distance=1cm, auto]
  \node[ punkt] (q0) {$q_0$};
  \node[punkt, inner sep=3pt,below=1.5cm of q0] (q1) {$q_1$};
  \node[punkt, inner sep=3pt,right=2cm of q1] (q2) {$q_2$};
  \node[punkt, inner sep=3pt,right=2cm of q2] (q3) {$q_3$};

  \path (q0)   edge[pil, loop left]
                	node[anchor=left,left]{$0/1$} (q0);
\path (q0)    edge [ pil]
                	node[anchor=north,left]{$1/1$} (q1);
\path (q1)    edge [ pil]
                	node[anchor=north,above]{$0/0$} (q2);
\path (q2)    edge [ pil]
                	node[anchor=north,above]{$0/0$} (q3);
\path (q3)    edge [ pil, loop above]
                	node[anchor=north,above]{$0/1$} (q3);

\end{tikzpicture}
\caption{A reduced candidate implementation $N$.}
\label{fig:fsm-counter}
\end{figure}
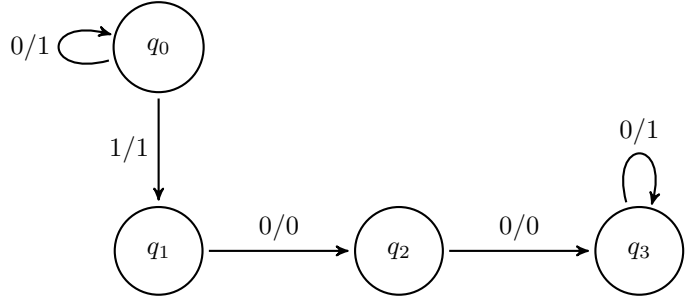

Now we can easily check that $M\yequ_T N$ because $\lambda(s_0,0000)=1111=\tau(q_0,0000)$ and $\lambda(s_0,100)=100=\tau(q_0,100)$.
But $M\ynequ N$ since we have $\lambda(s_0,1000)=1000\neq1001=\tau(q_0,1000)$. 
It is also easy to verify that $U(s_0)\subseteq U(q_0)$. 
We conclude that $T$ is not $4$-complete for $M$, and so it is also not $8$-complete for $M$,
where $8$ is the bound specified by Theorem~\ref{theo:bound-completeness}. 


\section{$m$-Perfectness}\label{m-perfectness}

Combining Theorem~\ref{theo:bound-completeness} and Corollary~\ref{cor:perfec-complete}(1),
we see that no test suite $T$ can be perfect for a given specification $M$ if we allow the
number of states in implementations to be put under test to grow beyond a bound
$k|S|$, where $|S|$ is the number of states in $M$ and $k$ is a constant that depends 
on $T$ alone.
This leads us to the notion of $m$-perfectness.
\begin{defin}\label{def:m-perfect}
Let $M$ be a FSM and $T$ be a test suite for $M$.
Then $T$ is \emph{$m$-perfect for $M$} iff for any FSM $N$ with at most $m$ states, if $M\ynalike N$ then $M\ynalike_T N$. 
\end{defin}
That is, $m$-perfectness guarantees that any difference in behavior between the specification $M$ and a implementation $N$ will be detected when we run the tests in $T$, even in the presence of blocking test cases, given that implementations are restricted to have at most $m$ states.
In other words, if $T$ is a $m$-perfect test suite for a specification $M$, then for any implementation under test $N$, if $M$ and $N$ are unlike, then they are also $T$-unlike, provided that $N$ has at most $m$ states.

We proceed to obtain necessary and sufficient conditions for $m$-perfectness, by showing that
a result analogous to Theorem~\ref{necessity-sufficiency-block}.
The following result will be useful when we consider certain bi-similarities. 
 \begin{lemma}\label{alike-walk}
 Let $M$ and $N$  be FSMs.  
 Let $n\geq 1$, $s_i\in S$, $p_i\in Q$ ($1\leq i\leq n$) and $x_i\in\yia$,  $a_i\in\yoa$, $b_i\in\yoa'$
 ($1\leq i< n$) be such that $s_i\ydio{x_i}{a_i} s_{i+1}$ and 
 $p_i\ydio{x_i}{b_i} p_{i+1}$ ($1\leq i< n$).
 Assume further that $s_1\yalike p_1$.
 Then $s_i\yalike p_i$ ($1\leq i\leq n$) and $a_1a_2\cdots a_{n-1}=b_1b_2\cdots b_{n-1}$.
 \end{lemma}
 \begin{proof}
 Let $\sigma=x_1x_2\cdots x_{n-1}$, $\omega_1=a_1a_2\cdots a_{n-1}$ and 
 $\omega_2=b_1b_2\cdots b_{n-1}$.
We clearly have  $s_{1}\ydio{\sigma}{\omega_1} s_n$ and $p_{1}\ydio{\sigma}{\omega_2} p_n$.
Definition~\ref{disting-equiv-block} immediately gives $\omega_1=\omega_2$, because $s_1\yalike p_1$ and $\sigma\in U(s_1)\cap U(q_1)$.
 
To see that $s_i\yalike p_i$ ($1\leq i\leq n$) we  go by induction on $n$. 
The basis follows from the hypothesis, and we proceed with the induction step.
Let $1\leq k<n$ and assume $s_k\yalike p_k$.
Let $\alpha=x_1\cdots x_k$.
Clearly $\delta(s_1,\alpha)=s_{k+1}$, $\mu(p_1,\alpha)=p_{k+1}$ and so 
$\alpha\in U(s_1)\cap U(p_1)$.
For te sake of contradiction, assume that $s_{k+1}\not\yalike p_{k+1}$.
By Definition~\ref{disting-equiv-block} we have two cases.
\begin{list}{}{\setlength{\leftmargin}{3ex}\setlength{\labelwidth}{1ex}}
\item[\hspace*{-1ex}\sc Case 1:] $U(s_{k+1})\ysyd U(p_{k+1})\neq \yemp$. 

\noindent Let $\beta\in U(s_{k+1})$ and $\beta\not\in U(p_{k+1})$.
This gives $\alpha\beta\in U(s_{1})$ and $\alpha\beta\not\in U(p_{1})$.
Hence $U(s_{1})\ysyd U(p_{1})\neq \yemp$, contradicting $s_1\yalike p_1$.
The situation when $\beta\not\in U(s_{k+1})$ and $\beta\in U(p_{k+1})$ is 
entirely analogous.

\item[\hspace*{-1ex}\sc Case 2:] $\beta\in U(s_{k+1})\cap U(p_{k+1})$ and
$\lambda(s_{k+1},\beta)\neq \tau(p_{k+1},\beta)$, for some 
$\beta\in\yias$.

\noindent This gives $\alpha\beta\in U(s_{1})\cap U(p_{1})$.
Moreover,
\begin{align*}
\lambda(s_1,\alpha\beta)&=\lambda(s_1,\alpha)\lambda(\delta(s_1,\alpha),\beta))=
\lambda(s_1,\alpha)\lambda(s_{k+1},\beta), \; \text{and}\\
\tau(p_1,\alpha\beta)&=\tau(p_1,\alpha)\tau(\mu(p_1,\alpha),\beta))=
\tau(p_1,\alpha)\tau(p_{k+1},\beta).
\end{align*}
Because $|\lambda(s_1,\alpha)|=|\tau(p_1,\alpha)|$ and $\lambda(s_{k+1},\beta)\neq \tau(p_{k+1},\beta)$, we get $\lambda(s_1,\alpha\beta)\neq \tau(p_1,\alpha\beta)$.
Since $\alpha\beta\in U(s_{1})\cap U(p_{1})$, this contradicts $s_1\yalike p_1$.
\end{list}
The proof is complete.\yfim
\end{proof}

The next result guarantees the existence of bi-simulations in the presence of blocking test cases. 
\begin{lemma}\label{alike-bisimulation}
Let $T$ be a $m$-perfect test suite for a  FSM $M$. 
Let $N$ be a FSM with at most $m$ states such that $M\yalike_T N$.
Then $M$ and $N$ are bi-similar. 
\end{lemma}
\begin{proof}
Define a relation $R_1\ysse S\times Q$ by letting
$(s,q)\in R_1$ if and only if $\delta(s_0,\alpha)=s$ and $\mu(q_0,\alpha)=q$ for some $\alpha \in \yias$, $s\in S$ and $q\in Q$.
Since $\delta(s_0,\yeps)=s_0$ and $\mu(q_0,\yeps)=q_0$ we get $(s_0,q_0)\in R_1$.

Now assume $(s,q)\in R_1$ and let $s\ydio{x}{a}r$ for some $r\in S$, $x\in \yia$
and $a\in \yoa$.
Since $(s,q)\in R_1$, the definition of $R_1$ gives some $\alpha \in \yias$ such that $\delta (s_0,\alpha)=s$ and $\mu(q_0,\alpha)=q$.
Composing, we get $\delta (s_0,\alpha x)=\delta (s, x)=r$ and so $\alpha x \in U(s_0)$.
Since $T$ is $m$-perfect for $M$ and $M\yalike_T N$, Definition~\ref{def:m-perfect} gives $M\yalike N$, that is $s_0\yalike q_0$.
Further, Definition~\ref{disting-equiv-block} and Remark~\ref{rema:equal-Us} imply 
$U(s_0)= U(q_0)$, and so $\alpha x\in U(q_0)$.
Then $\mu(q, x)=p$, for some $p\in Q$. 
Since $s_0\yalike q_0$, $\delta (s_0,\alpha)=s$ and $\mu(q_0,\alpha)=q$, Lemma~\ref{alike-walk} gives $s\yalike q$.
But $x\in U(s)\cap U(q)$, and so we must have $a=\lambda(s,x)=\tau(q,x)$.
Thus, we have found $p\in Q$ with $q\ydio{x}{a} p$.
Since $\delta(s_0,\alpha x) =r$ and $\mu(q_0,\alpha x)=p$, we also have
$(r,p)\in R_1$. This shows that $R_1$ is a simulation relation.

A similar argument will show that $R_2\ysse Q\times S$, where $R_2=R_1^{-1}$, is also a simulation relation.
Thus $M$ and $N$ are bi-similar, as desired. \yfim
\end{proof}

We now show the converse, that is, if $M$ is bi-similar to any FSM $N$ with at most $m$ states that is 
$T$-alike to it, then $T$ is a $m$-perfect test suite for $M$.
\begin{lemma}
\label{bisimulation-m-perfect}
Let $M$ be a FSM, $T$ a test suite for $M$, and $m\geq 1$. 
Assume that any FSM that is $T$-alike to $M$ with at most $m$ states is bi-similar to it.
Then $T$ is $m$-perfect for $M$.
\end{lemma}
\begin{proof}
We proceed by contradiction.
Assume that $T$ is not $m$-perfect for $M$.
Then, by Definition~\ref{def:m-perfect}, there exists a  FSM $N$ with at most $m$ states such that  $M\yalike_T N$ and $M\ynalike N$.
Hence, since $M\yalike_T N$, by Theorem~\ref{necessity-sufficiency-block} we know that $N$ is bi-similar to $M$, and so we have simulation relations $R_1\ysse S\times Q$ and 
$R_2\ysse Q\times S$.

Since $M\ynalike N$, by Definition~\ref{disting-equiv-block} we have two cases:
\begin{list}{}{\setlength{\leftmargin}{3ex}\setlength{\labelwidth}{1ex}}
\item[\hspace*{-1ex}\sc Case 1:] $\alpha\in U(s_0)\ysyd U(q_0)$, for some $\alpha\in \yias$. 

\noindent We may assume that $|\alpha|$ is minimum.
If $\alpha\in U(q_0)$ and $\alpha\not\in U(s_0)$, then we may write
$\alpha=\beta x$, where $\beta\in \yias$, $x\in\yia$ are such that $\beta\in U(q_0)\cap U(s_0)$.
Thus, $\delta(s_0,\beta)=s$, $\mu(q_0,\beta)=q$ and $\mu(q,x)=p$, for some $s\in S$ and some $q, p\in Q$.
Since $(q_0,s_0)\in R_2$, we can use Lemma~\ref{alike-walk} and write $(q,s)\in R_2$.
Because $R_2$ is a simulation and $\mu(q,x)=p$ we get some $r\in S$ such that $\delta(s,x)=r$.
But this gives $\delta(s_0,\alpha)=\delta(s_0,\beta x)=\delta(s,x)=r$, that is 
$\alpha\in U(s_0)$, a contradiction.
When $\alpha\not\in U(q_0)$ and $\alpha\in U(s_0)$, the argument is analogous.

\item[\hspace*{-1ex}\sc Case 2:] There is some $\alpha\in U(s_0)\cap U(q_0)$ 
with $\lambda(s_0,\alpha)\neq \tau(q_0,\alpha)$.

\noindent Again, assume that $|\alpha|$ is minimum.
Then, there are $\beta\in\yias$, $x\in\yia$, $s\in S$ and $q\in Q$ such that
$\alpha=\beta x$ and 
$\delta(s_0,\beta)=s$, $\mu(q_0,\beta)=q$.
Further, we get some $r\in S$, $p\in Q$ such that 
$\delta(s,x)=r$, $\mu(q,x)=p$, and $a=\lambda(s,x)\neq\tau(q,x)=b$.
Using the Lemma~\ref{alike-walk}, we may write $(s,q)\in R_1$.
Because we have $s\ydio{x}{a} r$ in $M$ and $R_1$ is a simulation, we know that there is some $t\in Q$
such that $q\ydio{x}{a} t$ in $N$, with $(r,t)\in R_1$.
But we already had $q\ydio{x}{b} p$ in $N$.
Hence, since $N$ is deterministic, we conclude that $a=b$, which is
a contradiction.
\end{list}
The proof is now complete. \yfim
\end{proof}

Combining the previous results we obtain necessary and sufficient conditions for $m$-perfectness. 
\begin{theorem}
\label{necessity-sufficiency-perfect}
Let $M$ be a FSM, $T$ be a test suite for $M$, and $m\geq 1$. 
Then $T$ is $m$-perfect for $M$ iff any $T$-alike FSM with at most $m$ states is bi-similar to $M$. 
\end{theorem}
\begin{proof}
Assume that $T$ is $m$-perfect for $M$.
Lemma~\ref{alike-bisimulation} guarantees that $N$ and  $M$ are bi-similar when $N$ is $T$-alike to $M$.
Now assume that any $T$-alike FSM with at most $m$ states is bi-similar to $M$. 
In this case, Lemma~\ref{bisimulation-m-perfect} guarantees that $T$ is $m$-perfect for $M$.  \yfim
\end{proof}

\section{Conclusions}\label{conclusion}

In this work we have studied the notion of test suite perfectness,
a notion similar to the classical one of test suite completeness, but now we may have the presence of so called  blocking test cases, that is, test cases that may not run to completion
either in the specification or in implementation models.
An accompanying notion of $p$-reduction was also introduced, similar to the classical notion of reduction  in FSMs.

We showed that any FSM can be $p$-reduced while maintaining the perfectness property, when
it was already present in the original FSM.
Using this result, we then proved that when the specification model and implementations to be put under test are both $p$-reduced, then perfectness can be characterized in terms of an isomorphism between both models.

We then established the relationship between perfectness and the classical notion of completeness.
We showed that perfectness is a strictly stronger relation, for specifications models of any sizes.
We then showed that when testing for perfectness one has to impose a limit on the number of states of the implementation models that are put under test.
This result was a consequence of a similar bound of the form $kn$ that we showed
must be imposed on the size of implementations when also testing for the classical notion
of $n$-completeness.
Here, $k$ is a constant that depends only on the test suite and $n$ is the number os states in 
the specification model.

We then characterized the $m$-perfectness property by establishing a necessary and sufficient
condition on the implementation models that are put under test, given a test suite and a specification model.
 
For future studies, we mention developing and testing a practical algorithm for testing
$m$-perfectness. 
Further, it may be the case that one can obtain tighter bounds on the size of implementation
models when testing for either $m$-perfectness or for $n$-perfectness.

\end{document}